%% file: main.tex
\documentclass[11pt]{article}
\usepackage[utf8]{inputenc}
\usepackage[american]{babel}
\usepackage[textwidth=6.5in,textheight=9in]{geometry}
\usepackage{setspace}
\usepackage{amsmath,amssymb,amsthm}
\usepackage{natbib}
\usepackage{bbm}
\usepackage[colorlinks=true,citecolor=blue]{hyperref}
\usepackage{graphicx}
\usepackage{xcolor}
\usepackage{authblk}
\newtheorem{assumption}{Assumption}
\newtheorem{proposition}{Proposition}
\newtheorem{theorem}{Theorem}

\newtheorem{lemma}{Lemma}

\theoremstyle{remark}
\newtheorem{remark}{Remark}

\theoremstyle{definition}

\begin{document}

\title{\textbf{Efficient estimation of cumulative incidence curves via data fusion with surrogates: application to integrated analysis of vaccine trial and immunobridging data}}


\author[1]{Pan Zhao}
\author[2]{Peter B. Gilbert}
\author[3]{Oliver Dukes}
\author[2]{Bo Zhang\thanks{Assistant Professor of Biostatistics, Vaccine and Infectious Disease Division, Fred Hutchinson Cancer Center. Email: {\tt bzhang3@fredhutch.org}. }}

\affil[1]{Statistical Laboratory, University of Cambridge}
\affil[2]{Vaccine and Infectious Disease Division, Fred Hutchinson Cancer Center}
\affil[3]{Department of Mathematics, Computer Science and Statistics, Ghent University}

\date{}

\maketitle
\begin{abstract}
Refined vaccine regimens containing variant-matched inserts are often authorized based on historical phase 3 efficacy trials together with immunobridging studies. Phase 3 trials are essential for establishing immune biomarkers that reliably predict disease risk or vaccine efficacy against clinical endpoints. Once such immune correlates are identified, updated vaccine regimens can be approved through immunobridging designs that compare the immunogenicity of the updated regimen to that of an already-approved vaccine. We develop methods of inference for the counterfactual cumulative incidence curve using participant-level data from both a historical vaccine efficacy trial and an immunobridging study. We further extend these methods to pathogens with multiple serotypes---such as dengue virus and influenza---by estimating cause-specific cumulative incidence curves. We describe the identification assumptions, propose efficient and multiply robust estimators, and assess their finite-sample performance through simulation studies. We then apply the proposed methods to (1) estimating the hypothetical cumulative incidence curve for a bivalent mRNA booster and (2) testing a key assumption of no controlled direct effects, using data from the COVID-19 Variant Immunologic Landscape (COVAIL) Trial, a multistage randomized clinical study evaluating the safety and immunogenicity of a second COVID-19 booster dose.
\end{abstract}


\input{1.Introduction}

\input{2.Notation}
\input{3.Identification}
\input{4.Estimation-Inference}
\input{5.Competing-risk}

\input{6.Simulation}
\input{7.Case-Study}
\input{8.Discussion}

\section*{Acknowledgements}
We thank the participants, site staff, and investigators of the COVAIL trial sponsored by the Division of Microbiology and Infectious Diseases at the National Institute Of Allergy and Infectious Diseases of the National Institutes of Health (NIAID NIH). We thank Dr. Nadine Rouphael and Dr. Angela Branche for helpful suggestions. This work is funded by the NIAID NIH (R01AI192632 to BZ, OD, and PBG and R37AI054165 to PBG). The content is solely the responsibility of the authors and does not necessarily represent the official views of the National Institutes of Health.

{
\singlespacing
\bibliographystyle{plainnat}
\bibliography{Bibliography}
}

\clearpage
\pagenumbering{arabic}
\appendix

\begin{center}
{\large\bf Supplemental Materials to ``Efficient estimation of cumulative incidence curves via data fusion with surrogates: application to integrated analysis of vaccine trial and immunobridging data" by Pan Zhao, Peter B. Gilbert, Oliver Dukes, and Bo Zhang. }
\end{center}

\input{SM.Proofs}

\input{SM.Data}

\end{document}

%% file: 1.Introduction.tex
\section{Introduction}
\label{sec: intro}

\subsection{From immune correlates of risk and protection to immunobridging}
\label{subsec: CoR and CoP; bridging}
An immune correlate of risk (CoR) is an immune biomarker, typically measured after vaccination, that can predict the risk of a clinical outcome, such as time to symptomatic infection, in a defined cohort. In contrast, an immune correlate of protection (CoP) is an immune biomarker that predicts vaccine efficacy (VE) relative to placebo, or the relative vaccine efficacy (relative VE) between two vaccines. During the COVID-19 pandemic, a concerted national and international research effort established neutralizing antibody (nAb) titers as a credible CoP for COVID-19. This conclusion was supported by a series of harmonized phase 3 clinical trials \citep{gilbert2022immune,fong2022immune,fong2023immune,benkeser2023immune}, other vaccine efficacy trials \citep{feng2021correlates}, and meta-analyses of efficacy trials and observational studies \citep{khoury2021neutralizing}. \citet[Figure 1]{gilbert2022covid} illustrates the estimated VE as a function of nAb titers measured at an approximately peak timepoint, derived from five major placebo-controlled phase 3 trials conducted early in the pandemic when the Ancestral Wuhan-Hu-1 strain predominated \citep{gilbert2022immune}. Under causal identification assumptions described in \citet{gilbert2023controlled}, each point along the curve represents a causal contrast in COVID-19 risk between (1) receiving a specific vaccine regimen that elicits a given peak nAb titer level and (2) receiving placebo.

Identifying reliable immune correlates is a central goal in vaccinology \citep{plotkin2012nomenclature}, with myriad applications, including two key ones. First, validated immune correlates help elucidate the mechanisms of vaccine-induced protection and guide the development of new vaccines. Second, although randomized controlled trials remain the gold standard, they are not always feasible or timely. From a regulatory perspective, immune correlates provide a pathway for authorizing or approving (1) existing vaccines in different populations and (2) updated or reformulated vaccines, such as those with variant-matching inserts, based on immunobridging studies rather than full-scale efficacy trials. Once an immune marker, such as nAb titer, is validated as a CoP for a specified context of use, the same vaccine or an updated vaccine may be approved through immunobridging studies. In these studies, post-vaccination immune markers, possibly against a new circulating strain, elicited by the investigational vaccine in the target population are compared to those elicited by a previously approved vaccine. Compared to conducting a new phase 3 efficacy trial, immunobridging offers a faster and more cost-effective regulatory pathway.

In this article, we study how to identify and efficiently estimate the counterfactual cumulative incidence of an updated vaccine regimen, as well as the relative vaccine efficacy of an investigational vaccine compared to an approved one, in a target population. In principle, this question could be addressed---with appropriate uncertainty quantification---by conducting a double-blind, randomized, active-controlled trial in the target population. However, such trials may not be practical. Instead, our goal is to answer this question using two sources of data:
(i) a historical clinical trial dataset ($\mathcal{D}_h$) containing participant-level information on baseline covariates, immune marker levels, and a time-to-event endpoint (e.g., time to virologically confirmed symptomatic infection); and
(ii) an immunobridging study dataset ($\mathcal{D}_b$) containing participant-level data on baseline covariates and immune marker levels, possibly against a new variant, in the target population.

\subsection{Generalizing the risk and efficacy with post-randomization events}
\label{subsec: intro generalize ITT}
The statistical problem of estimating the counterfactual cumulative incidence curve using historical clinical trial data and immunobridging study data falls within the growing literature on the generalizability and transportability of causal effects \citep{cole2010generalizing,dahabreh2022generalizing,colnet2024causal,degtiar2023review,graham2024towards}. Unlike conventional generalizability problems, which primarily address differences in baseline covariate distributions between populations, our setting introduces an additional layer of complexity: we observe data on a post-randomization variable---namely, the immune marker level at a specific study visit post-vaccination---in both the historical trial and the immunobridging datasets.

At least four sources of heterogeneity must be accounted for when estimating cumulative incidence or relative VE based on one or more post-randomization biomarkers. First, vaccine-induced immune responses may vary across individuals. Even among participants with similar baseline characteristics, the same vaccine platform typically elicits heterogeneous immune marker levels due to factors such as age, sex, host genetics, microbiome composition, and prior infections or vaccinations \citep{huang2022baseline}. Second, baseline characteristics may correlate with clinical risk, even when individuals have comparable post-vaccination immune marker levels. For example, antigen-na\"ive versus non-na\"ive individuals may have different risks of symptomatic infection despite similar post-vaccination immune marker levels. Third, the target population, from which the immunobridging dataset ($\mathcal{D}_b$) is drawn, may differ from the historical trial population in the distribution of baseline covariates. Fourth, the distribution of circulating pathogen strains---or, more broadly, the external force of infection---may vary across populations or over time.

A rigorous immunobridging framework must account for all these sources of heterogeneity. Additionally, it must quantify statistical uncertainty arising from sampling variability in both the historical trial dataset ($\mathcal{D}_h$) and the immunobridging dataset ($\mathcal{D}_b$).

\subsection{Related work, our contributions, and three immunobridging tasks}
\label{subsec: intro literature and contribution}
We briefly review prior work on generalizing and transporting causal effects, focusing specifically on methods that address a post-randomization event. \citet{rudolph2016robust} and \citet{he2024generalizing} study generalizability under an instrumental variable framework, considering a binary instrumental variable (treatment assignment) and a binary treatment received in the context of a randomized trial with noncompliance. \citet{gilbert2022novel} propose mediation estimators that evaluate the extent to which differences in effects from various sources can be explained by compositional factors and the mechanisms generating mediating or intermediate variables, using mediator and outcome data obtained from multiple external sources.

Prior work most closely related to this article includes \citet{gilbert2016predicting}, \citet{athey2025surrogate}, and \citet{gilbert2024surrogateendpointbasedprovisional}. \citet{gilbert2016predicting} introduced a transport formula and corresponding plug-in estimator to estimate clinical endpoint risk and vaccine efficacy in a new population, using a combination of historical clinical trial data and immunobridging data. This approach was later applied in \citet{gilbert2019bridging} to predict the vaccine efficacy of a tetravalent dengue vaccine when immunobridging from children and adolescents to adults.

Both \citet{athey2025surrogate} and \citet{gilbert2024surrogateendpointbasedprovisional} study settings in which baseline covariates and post-randomization surrogate endpoints are observed in both an experimental sample and an observational sample, but the long-term clinical outcome is observed only in the observational sample. Their goal is to estimate a causal effect (e.g., the average treatment effect in \citet{athey2025surrogate} and a more general causal contrast in \citet{gilbert2024surrogateendpointbasedprovisional}) for the population underlying the experimental sample—without observing the long-term outcome in that sample. A key distinction between the two is that \citet{athey2025surrogate} assume that treatment assignment is not observed in the observational sample, whereas \citet{gilbert2024surrogateendpointbasedprovisional} allow the treatment indicator to be observed in both datasets. The latter setup matches ours in the current work if the observational sample is taken to be the vaccine arm of the historical phase 3 trial that supported the vaccine’s approval, and the experimental sample is taken to be both the investigational and approved vaccine arms of the immunobridging trial.

This work extends the frameworks of \citet{athey2025surrogate} and \citet{gilbert2024surrogateendpointbasedprovisional} to time-to-event (survival) endpoints subject to informative right-censoring, with treatment indicators always observed. Our primary estimands are the counterfactual cumulative incidence curve and the relative vaccine efficacy curve, each indexed by study time. This extension is practically important, as survival outcomes---such as time to virologically confirmed symptomatic infection---remain the primary endpoints in most vaccine efficacy trials. We further consider cause-specific cumulative incidence curves to address scenarios in which multiple pathogen serotypes co-circulate. 

The methods developed in this article address three distinct immunobridging tasks.
\begin{description}
       \item[Task I:] A vaccine is approved based on a phase 3 clinical trial in one population. During the trial, only one pathogen strain was circulating. Researchers then aim to estimate the effectiveness of an investigational vaccine---built on the same platform but potentially eliciting stronger immune responses---against the same pathogen strain in a target population (which may differ from the historical trial population), using the phase 3 trial data together with an additional immunogenicity study conducted in the target population.
       \item[Task II:] A vaccine is approved following positive phase 3 trial results in a given population against a single viral strain (e.g., the ancestral SARS-CoV-2 strain) that was circulating. Researchers then seek to evaluate the effectiveness of a new variant-matching vaccine in the target population where a new single strain (e.g., Omicron BA.4/5) is circulating, using the historical phase 3 trial data together with an additional immunogenicity trial that measures immune responses to the new variant and is conducted in the target population. 
       
       \item[Task III:] Many pathogens, such as dengue and influenza, circulate with multiple serotypes concurrently. Researchers seek to estimate cause-specific and overall counterfactual cumulative incidence curves for a new investigational vaccine in a target population, using data from a historical VE trial and an immunobridging trial. In this setting, the same set of strains circulates in the historical VE trial and the immunobridging trial. The historical VE trial records serotype-specific time-to-event endpoints and immune responses to the approved vaccine across all co-circulating serotypes, whereas the immunobridging trial measures immune responses to both the approved and investigational vaccines across the same set of co-circulating serotypes.
\end{description}

Figure \ref{fig:scheme} summarizes the study setting and key design features of the three tasks.

\begin{figure}[ht]
    \centering
\includegraphics[width=0.99\linewidth, trim=0cm 4.2cm 0cm 5cm, clip]{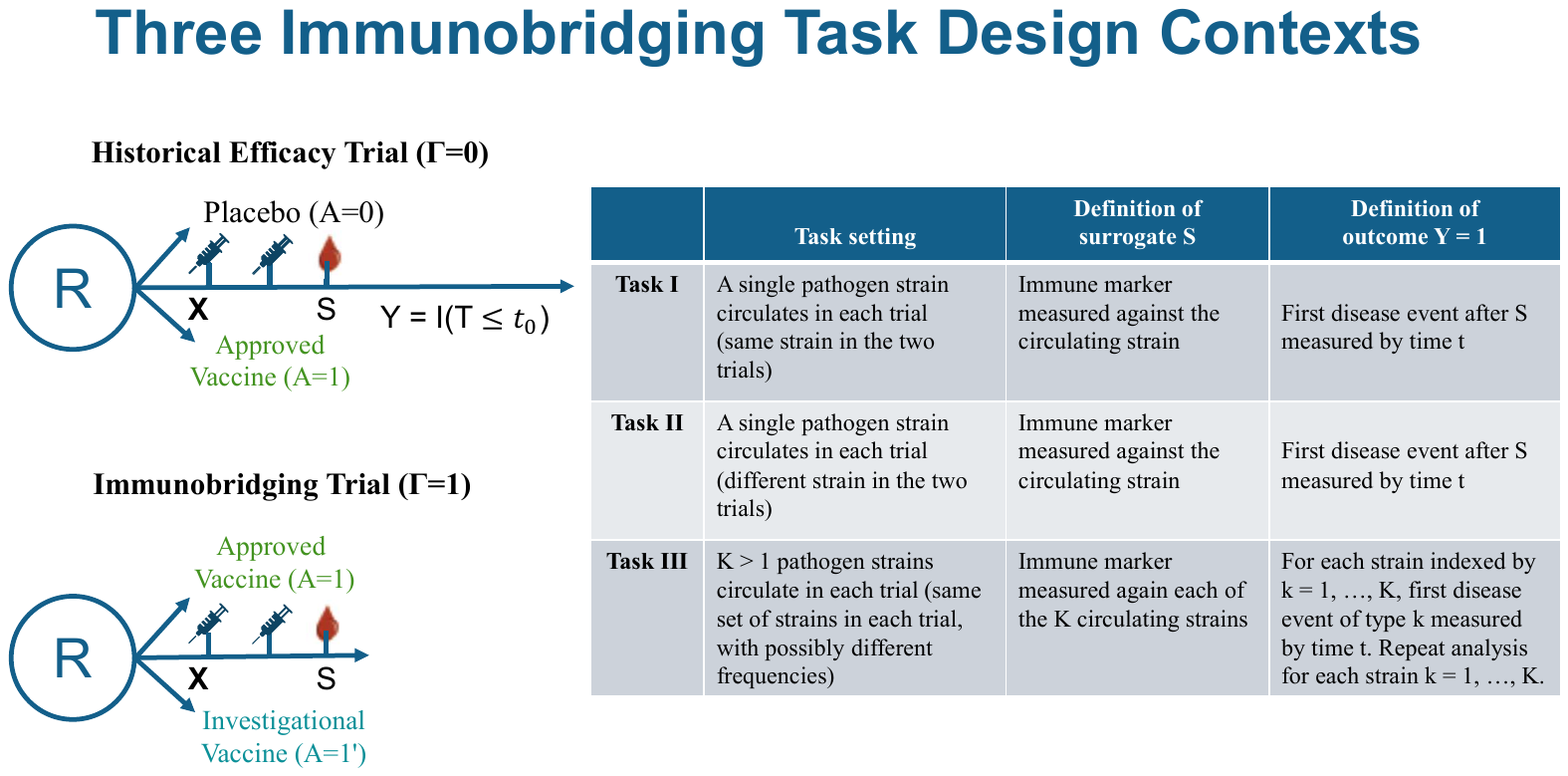}
    \caption{Three Immunobridging Task Design Contexts}
    \label{fig:scheme}
\end{figure}

\subsection{COVID-19 Variant Immunologic Landscape (COVAIL) trial}
\label{subsec: intro COVAIL}
The methods developed in this article can be applied to a broad range of scientific questions in immunology and vaccinology. Here, we focus on the COVID-19 Variant Immunologic Landscape (COVAIL) trial, a randomized clinical study conducted across 22 sites in the United States. The COVAIL trial used a multistage design in which participants were sequentially enrolled and randomized within each stage to receive one of several variant vaccines administered as a second booster dose. The primary objective of the trial is to evaluate the immunogenicity and safety of variant vaccines from three platforms: Moderna mRNA (Stage 1), Pfizer-BioNTech mRNA (Stages 2 and 4), and Sanofi recombinant protein (Stage 3) \citep{branche2023comparison,branche2023immunogenicity}. Recent studies have shown that post-vaccination neutralizing antibody (nAb) levels are inversely correlated with time to COVID-19, with stronger correlations among participants with prior SARS-CoV-2 infection \citep{fong2025neutralizing,zhang2025neutralizing}. The case study in this article focuses on Stages 2 and 4.

\subsection{Organization}
\label{subsec: intro organization}
The remainder of the paper is organized as follows. Section \ref{sec: statistical framework} introduces the statistical framework and notation. Section \ref{sec: identification} outlines the identification assumptions and presents the main identification results. Section \ref{sec: estimation and inference} proposes the estimators and establishes their theoretical properties. Section \ref{sec: extension to competing risk} extends the methods to settings with multiple serotype-specific time-to-event endpoints. Section \ref{sec: simulation} reports results from simulation studies, and Section \ref{sec: case study COVAIL} applies the proposed methods to data from the COVAIL trial. Section \ref{sec: discussion} concludes with a discussion. Proofs and additional technical details are provided in the Supplemental Materials. An \textsf{R} implementation of the proposed methods is publicly available at \url{https://github.com/panzhaooo/data-fusion-surrogates-survival-analysis}.

%% file: 2.Notation.tex
\section{Statistical Framework}
\label{sec: statistical framework}
We present the statistical framework for Task I in Sections \ref{subsec: data} and \ref{subsec: estimand}. We then extend this framework to Task II in Section \ref{subsec: notation task II}. Finally, we develop the framework for Task III, which accommodates the co-circulation of multiple serotypes, in Section \ref{sec: extension to competing risk}.

\subsection{Data structure and potential outcomes}
\label{subsec: data}

Consider a historical randomized, placebo-controlled clinical trial dataset $\mathcal{D}_{h} = \{(\boldsymbol{X}_i, A_i, S_i, T_i): i = 1, \ldots, n_h\}$, where  $\boldsymbol{X}_i \in \mathcal{X}$ denotes a vector of baseline covariates (e.g., age, baseline immunity factors, and other risk factors), $A_i \in \{0, 1\}$ denotes vaccine assignment ($A_i = 1$ for vaccine; $A_i = 0$ for placebo), $S_i \in \mathcal{S}$ is an immune marker measured at a prespecified post-vaccination visit (e.g., $15$ days after vaccination), and $T_i$ represents the time from when the immune marker $S$ is measured to the clinical endpoint of interest (e.g., COVID-19). For simplicity, we assume that all participants attend the visit at which the immune marker is measured and that no participants experience the clinical endpoint prior to that visit.

In addition to $\mathcal{D}_{h}$, researchers also have access to an immunobridging dataset $\mathcal{D}_{b} = \{(\boldsymbol{X}_i, A_i, S_{i}) : i = n_h + 1, \ldots, n_h + n_b\}$, where $A_i \in \{1, 1'\}$ indicates randomized assignment to the approved vaccine regimen ($A_i = 1$), identical to that in the historical efficacy trial, or an investigational vaccine regimen ($A_i = 1'$), and $S_i$ denotes the same immune marker measured at a prespecified post-vaccination study visit. Unlike the historical trial dataset $\mathcal{D}_h$, the immunobridging dataset $\mathcal{D}_b$ does not include time-to-event outcomes. In addition, $\mathcal{D}_b$ typically has a much smaller sample size than $\mathcal{D}_h$ (e.g., several hundred compared with tens of thousands). Let $\mathcal{P}_h$ and $\mathcal{P}_b$ denote the marginal distributions of baseline covariates $\boldsymbol{X}$ in the historical trial and immunobridging populations, respectively. To track study membership, we define a binary indicator $\Gamma_i$ such that $\Gamma_i = 0$ if participant $i$ is from the historical trial ($\mathcal{D}_h$) and $\Gamma_i = 1$ if from the immunobridging study ($\mathcal{D}_b$).

We adopt the Neyman–Rubin potential outcomes (PO) framework to formalize our problem \citep{neyman1923application,rubin1974estimating}. Let $S(A = a)$ denote the potential immune marker level measured at a specific post-vaccination visit had a participant received vaccine regimen $A = a$. In defining $S(A = a)$, we assume Rubin's Stable Unit Treatment Value Assumption (SUTVA) \citep{rubin1980discussion}, so that a participant's post-vaccination marker level depends only on their own treatment assignment and there is a single version of each treatment, ensuring that $S(A = a)$ is well-defined.

Both $\mathcal{D}_h$ and $\mathcal{D}_b$ arise from randomized trials, so the treatment assignment mechanism is known in each dataset. Each participant has associated potential outcomes $T(S = s, A = a)$, denoting the counterfactual time-to-event outcome (e.g., time to ancestral strain COVID-19) if the participant were assigned treatment $A = a$ and their post-vaccination immune marker $S$ were set to level $s$. It is important to distinguish this joint-intervention potential outcome, $T(s, a)$, from the potential outcome $T(S(a), a)$---or simply $T(a)$---which represents the time-to-event outcome under assignment to $A = a$ with $S$ taking its natural value under $A = a$, that is, $S(a)$.

\subsection{Estimands}
\label{subsec: estimand}
We first consider the quantity $\mathbb{E}[y(T(a,s)) \mid \boldsymbol{X}, \Gamma = 0]$, where $T(a,s)$ denotes the counterfactual time-to-event outcome under vaccine assignment $A=a$ and immune marker level $S=s$, and $y(\cdot)$ is a real-valued transformation. When $y(\cdot)$ is chosen as the indicator function $I\{~\cdot \le t_0\}$ for a prespecified time point $t_0$, the quantity $\mathbb{E}[y(T(a,s)) \mid \boldsymbol{X}, \Gamma = 0]$ reduces to the conditional controlled risk in the historical phase 3 trial, 
$P\{T(a, s) \leq t_0 \mid \boldsymbol{X}, \Gamma = 0\}$, as defined in \citet{gilbert2023controlled}.
This conditional risk parameter represents the probability of experiencing the clinical endpoint by time $t_0$ under a hypothetical joint intervention that simultaneously sets $A = a$ and $S = s$, conditional on baseline covariates.

In addition, we define the quantity $R(a; \boldsymbol{X}, \Gamma = 0) = \mathbb{E}[y(T(a)) \mid \boldsymbol{X}, \Gamma = 0]$, which, unlike $\mathbb{E}[y(T(a, s)) \mid \boldsymbol{X}, \Gamma = 0]$, is not a function of the post-vaccination marker $S = s$. Averaging $R(a; \boldsymbol{X}, \Gamma = 0)$ over the historical trial population yields the marginalized potential outcome: 
\begin{equation*}
    R(a; \Gamma = 0) = \mathbb{E}_{\mathcal{P}_h}\left[\mathbb{E}\left[y(T(a)) \mid \boldsymbol{X}, \Gamma = 0\right]\right],~a \in \{0, 1\}.
\end{equation*}
\noindent By choosing $y(\cdot) = I\{~\cdot \leq t_0\}$, $R(a; \boldsymbol{X}, \Gamma = 0)$ and $R(a; \Gamma = 0)$ reduce to the conditional and marginalized risk parameters $P\{T(a) \leq t_0 \mid \boldsymbol{X}, \Gamma = 0\}$ and $\mathbb{E}_{\mathcal{P}_h}\left[P\{T(a) \leq t_0 \mid \boldsymbol{X}, \Gamma = 0\}\right]$, respectively. Both $R(a; \boldsymbol{X}, \Gamma = 0)$ and $R(a; \Gamma = 0)$ are identifiable and can be estimated without bias, with appropriate uncertainty quantification, from the historical randomized controlled trial, under standard causal identification assumptions and suitable assumptions on the censoring mechanism (see, e.g., \citet{westling2024inference}).

Next, consider a hypothetical phase 3 trial conducted in the target population, evaluating an investigational vaccine targeting the same pathogen strain as in the historical VE trial. We define the conditional and marginalized potential time-to-event outcomes in this new setting as $R(a; \boldsymbol{X}, \Gamma = 1) = \mathbb{E}[y(T(a)) \mid \boldsymbol{X}, \Gamma = 1]$ and 
$R(a; \Gamma = 1) = \mathbb{E}_{\mathcal{P}_b}\left[\mathbb{E}\left[y(T(a)) \mid \boldsymbol{X}, \Gamma = 1\right]\right],~a \in \{1, 1'\}$, respectively. Our target parameter for Task I is the marginalized potential time-to-event outcome in the target population, $R(a; \Gamma = 1)$, with a particular focus on the marginalized risk $\mathbb{E}_{\mathcal{P}_b}\left[P\{T(a) \leq t_0 \mid \boldsymbol{X}, \Gamma = 1\}\right]$.

Importantly, differences in the amount of pathogen circulating and exposing people, or the immune response elicited by the approved versus investigational vaccines, imply that conditional mean exchangeability cannot generally be expected to hold for $R(A = 1; \boldsymbol{X}, \Gamma = 0)$ and $R(A = 1'; \boldsymbol{X}, \Gamma = 1)$, that is, 
$R(A = 1; \boldsymbol{X}, \Gamma = 0) \neq R(A = 1'; \boldsymbol{X}, \Gamma = 1).$ This breakdown of exchangeability mirrors challenges encountered in generalizing intention-to-treat effects in the presence of treatment noncompliance, where the treatment received differs from the treatment assigned; see \citet{he2024generalizing} for a detailed discussion.

Finally, we define the following relative vaccine efficacy comparing the investigational vaccine $A = 1'$ to the approved vaccine $A = 1$ in the target population:
\begin{equation*}
    \text{relVE}(1, 1'; t_0, \Gamma = 1) = 1 - \frac{R(A = 1'; \Gamma = 1)}{R(A = 1; \Gamma = 1)},
\end{equation*}
with $y(\cdot) = I\{~\cdot \leq t_0\}$ in $R(a;\Gamma = 1),~a \in \{1, 1'\}$.

\subsection{Extension to immunobridging with a variant-matching vaccine}
\label{subsec: notation task II}
In Task II, researchers seek to evaluate the effectiveness of an investigational vaccine updated to match the currently circulating strain. Let $\mathcal{D}_{h} = \{(\boldsymbol{X}_i, A_i, S_i, T_i): i = 1, \ldots, n_h\}$ be the same historical VE trial dataset described in Section \ref{subsec: data}. To address Task II, the immunobridging study now measures the immune marker against the circulating strain. Let $\mathcal{D}_{b} = \{(\boldsymbol{X}_i, A_i, S'_{i}) : i = n_h + 1, \ldots, n_h + n_b\}$ denote the corresponding immunobridging dataset, where $S'_i$ represents the same type of post-vaccination immune marker as $S_i$ in $\mathcal{D}_h$, but measured against the circulating strain. In the COVID-19 setting, $S_i$ may correspond to the approximate peak neutralizing antibody level against the ancestral SARS-CoV-2 strain, whereas $S'_i$ denotes the approximate peak nAb level against the Omicron BA.4/BA.5 strain.

Moreover, we define $T'(S' = s, A = a)$ as the counterfactual time-to-event outcome for the clinical endpoint caused by the currently circulating strain (e.g., Omicron BA.4/BA.5 COVID-19), had a participant received vaccine $A = a$ and their variant-specific immune marker level $S'$ set to $s$. This potential outcome $T'(S' = s, A = a)$ is defined for the immunobridging trial $\Gamma = 1$. We also define the counterfactual time-to-event outcome $T'(a)$, $R'(a; \boldsymbol{X}, \Gamma = 1) = \mathbb{E}[y(T'(a)) \mid \boldsymbol{X}, \Gamma = 1]$, and 
$R'(a; \Gamma = 1) = \mathbb{E}_{\mathcal{P}_b}\left[\mathbb{E}\left[y(T'(a)) \mid \boldsymbol{X}, \Gamma = 1\right]\right],~a \in \{1, 1'\}$, all defined against the currently circulating strain in the immunobridging trial. These are the target parameters for Task II.

%% file: 3.Identification.tex
\section{Identification for Tasks I and II}
\label{sec: identification}

\subsection{Task I}
\label{subsec: identification assumption Task I}

We first discuss assumptions needed to identify the target parameter $R(a; \Gamma = 1),~a \in \{1, 1'\},$ for Task I. Assumptions \ref{ass: standard assumpion in Dh}(i)--\ref{ass: standard assumpion in Dh}(iv) are standard assumptions in many causal inference problems. We assume Assumptions \ref{ass: standard assumpion in Dh}(i)--\ref{ass: standard assumpion in Dh}(iv) hold for the historical phase 3 clinical trial and a hypothetical trial comparing the investigational vaccine ($A = 1'$) versus the approved vaccine ($A = 1$).

\begin{assumption}\label{ass: standard assumpion in Dh}
{\rm (i. Consistency)} $S = S(A)$, $T = T(A, S)$ and $T = T(A) = T(A, S(A))$; {\rm (ii. Randomization)} $A \perp T(a,s) \mid \boldsymbol{X}$; {\rm (iii. Strong sequential ignorability)} $S \perp T(a,s) \mid \boldsymbol{X} = \boldsymbol{x}$; {\rm (iv. Positivity)} $P(A = a \mid \boldsymbol{X} = \boldsymbol{x}) > 0$ and $f(s \mid \boldsymbol{x}, a) > 0$, for all $a \in \{0, 1, 1'\}, \boldsymbol{x} \in \mathcal{X}, s \in \mathcal{S}$.
\end{assumption}

Assumption \ref{ass: standard assumpion in Dh}(iii) effectively assumes that 
$S$ is randomized within strata defined by vaccine arm and baseline covariates.

Assumption \ref{ass: overlap} is often required in the generalizability and transportability literature.

\begin{assumption}[Overlap]\label{ass: overlap}
The support of baseline covariates $\boldsymbol{X}$ in the target trial population is a subset of that in the historical trial population.
\end{assumption}

In practice, when the support of the target population lies outside that of the historical trial population, one pragmatic strategy is to restrict the analysis to the region of sufficient covariate overlap. This approach helps mitigate extrapolation beyond the support of the available data and ensures more credible inference; see, e.g., \citet[Section 6.2]{he2024generalizing} for an illustrative example.

Assumption \ref{ass: conditional exchangeability task I} is the key ``conditional exchangeability assumption" required to identify the target parameter $R(A = 1; \Gamma = 1)$. It extends the conventional exchangeability assumption \citep{dahabreh2022generalizing} to a hypothetical joint intervention that simultaneously sets $A = a$ and $S = s$.

\begin{assumption}[Conditional exchangeability of the approved vaccine in two trial settings]\label{ass: conditional exchangeability task I} Let $y(\cdot)$ denote a real-valued transformation. Conditional exchangeability is said to hold for the approved vaccine $A = 1$ in the historical trial context and the immunobridging trial context if
\begin{equation*}
\mathbb{E}\left[y(T(A = 1, S = s)) \mid \mathbf{X} = \boldsymbol{x}, \Gamma = 0\right] = \mathbb{E}\left[y(T(A = 1, S = s)) \mid  \mathbf{X} = \boldsymbol{x}, \Gamma = 1\right],
\end{equation*}
for $\boldsymbol{x} \in \mathcal{X}, s \in \mathcal{S}$ where $f(s \mid \boldsymbol{x}, 1) > 0$.
\end{assumption}

For infectious diseases, variability in the ``background intensity" of exposure and transmission over geography and calendar time poses a challenge to the conditional exchangeability assumption, which may be violated if the two trials differ in either factor. A concept addressing this challenge defines a background intensity covariate with participants' values assigned using an external surveillance data base (common to the two trials) that reports disease incidence by geography and calendar date (see, e.g., \citet{zhang2025neutralizing}). Consequently, Assumption~\ref{ass: conditional exchangeability task II} is more plausible when $\mathbf{X}$ includes variables that predict participants' exposure risk as well as relevant ecological factors.

Finally, Assumption \ref{ass: same joint intervention A = 1 & 1' Task I} is the other key assumption to identify $R(A = 1'; \Gamma = 1)$. It states an equivalence between the joint intervention $(A = 1, S = s)$ and $(A = 1', S = s)$ in the immunobridging-trial context. 

\begin{assumption}[No controlled direct effects comparing investigational and approved vaccines in the immunobridging trial]\label{ass: same joint intervention A = 1 & 1' Task I}
The no controlled direct effects assumption is said to hold for the approved vaccine $A = 1$, the investigational vaccine $A = 1'$, and the immune marker $S$ if the following holds: 
\begin{equation*}
\mathbb{E}\left[y(T(A = 1, S = s)) \mid \mathbf{X} = \boldsymbol{x}, \Gamma = 1\right] = \mathbb{E}\left[y(T(A = 1', S = s)) \mid  \mathbf{X} = \boldsymbol{x}, \Gamma = 1\right],
\end{equation*}
for $\boldsymbol{x} \in \mathcal{X}, s \in \mathcal{S}$ where $f(s \mid \boldsymbol{x}, a) > 0$, $a \in \{1, 1'\}$.
\end{assumption}

Under Assumption \ref{ass: standard assumpion in Dh}, Assumption \ref{ass: same joint intervention A = 1 & 1' Task I} then implies
\[
\mathbb{E}\left[y(T) \mid A = 1, S = s, \mathbf{X} = \boldsymbol{x}, \Gamma = 1\right] = \mathbb{E}\left[y(T) \mid A = 1', S = s, \mathbf{X} = \boldsymbol{x}, \Gamma = 1\right],
\]
which is a version of the surrogacy assumption described by \citet{athey2025surrogate} and \citet{gilbert2024surrogateendpointbasedprovisional}. We formulate this assumption using potential outcomes rather than observed data, because the former is more interpretable in the context of our case study.

Assumption \ref{ass: same joint intervention A = 1 & 1' Task I} will be violated if the investigational vaccine influences the clinical outcome through immune mechanisms not captured by $S$, and these mechanisms differ between the investigational and approved vaccines. 

\begin{remark}[Relaxing Assumptions \ref{ass: conditional exchangeability task I} and \ref{ass: same joint intervention A = 1 & 1' Task I}]\label{remark: relax Assumptions 3 & 4}
    Because the clinical endpoint is not observed in the immunobridging trial, neither Assumption \ref{ass: conditional exchangeability task I} nor Assumption \ref{ass: same joint intervention A = 1 & 1' Task I} yields testable implications. We propose two strategies to relax these assumptions. First, one can adopt a bounding perspective: instead of requiring equality, posit a directional inequality informed by domain knowledge. For example, rather than assuming Assumption \ref{ass: same joint intervention A = 1 & 1' Task I}, it may be more plausible to assume that the controlled direct effect of the investigational vaccine compared to the approved vaccine is positive. As a concrete example, in the vaccine setting, $S$ may represent neutralizing antibody levels, while the controlled direct effect captures other protective immune responses such as cell-mediated and non-neutralizing humoral responses. Under certain circumstances, candidate vaccines may induce these components in similar proportions \citep{krause2022making}; hence, it may be reasonable to assume that the investigational vaccine induces equivalent or improved cell-mediated and non-neutralizing humoral responses compared to the approved vaccine, provided its neutralizing antibody response is no worse. In this case, we have
\begin{equation*}
P\{T(A = 1, S = s) \leq t_0 \mid \mathbf{X}, \Gamma = 1\} \leq P\{T(A = 1', S = s) \leq t_0 \mid  \mathbf{X}, \Gamma = 1\}, 
\end{equation*}
where we have let $y(\cdot) = I\{~\cdot \leq t_0\}$. 

Alternatively, Assumptions \ref{ass: conditional exchangeability task I} and \ref{ass: same joint intervention A = 1 & 1' Task I} may be relaxed through sensitivity analysis. For example, to relax Assumption \ref{ass: same joint intervention A = 1 & 1' Task I}, one may introduce a sensitivity function $h(\mathbf{X})$ and specify
\begin{equation*}
    g(\mathbb{E}\left[y(T(A = 1, S = s)) \mid  \mathbf{X}, \Gamma = 1\right]) = g(\mathbb{E}\left[y(T(A = 1', S = s)) \mid  \mathbf{X}, \Gamma = 1\right]) + h(\mathbf{X}),
\end{equation*}
where $g(\cdot)$ is a link function (e.g., the logit link for a binary outcome, such as when $y(\cdot) = I(\cdot \leq t)$), and $h(\mathbf{X})$ quantifies the deviation from equality. Related sensitivity function–based approaches are discussed in \citet{brumback2004sensitivity}, \citet{he2024generalizing}, and \citet{gilbert2024surrogateendpointbasedprovisional}.
\end{remark}


Proposition \ref{prop:identification-complete-data} is analogous to the identification results in \citet{athey2025surrogate}, though the setting here is slightly different. Proposition \ref{prop:identification-complete-data} identifies the target parameter $R(a; \Gamma = 1),~a \in \{1, 1'\},$ for Task I in three distinct ways under Assumptions \ref{ass: standard assumpion in Dh}--\ref{ass: same joint intervention A = 1 & 1' Task I}.

\begin{proposition}\label{prop:identification-complete-data}
Under Assumptions \ref{ass: standard assumpion in Dh}--\ref{ass: same joint intervention A = 1 & 1' Task I}, for $a \in \{1, 1'\}$, we have
\begin{equation*}{\textstyle
R(a; \Gamma = 1) = \mathbb{E}\left[\frac{(1 - \Gamma) I\{A = 1\}}{\kappa} \frac{f(\Gamma = 1 \mid \boldsymbol{X})}{f(A = 1, \Gamma = 0 \mid \boldsymbol{X})} \frac{f(S \mid \boldsymbol{X}, A = a, \Gamma = 1)}{f(S \mid \boldsymbol{X}, A = 1, \Gamma = 0)} y(T)\right],}
\end{equation*}
where $\kappa = f(\Gamma = 1)$, 
\begin{equation*}{\textstyle
R(a; \Gamma = 1) = \mathbb{E}\left[\frac{\Gamma I\{A = a\}}{\kappa f(A = a \mid \boldsymbol{X}, \Gamma = 1)} \mu(\boldsymbol{X}, 1, S)\right],}
\end{equation*}
where $\mu(\boldsymbol{x}, 1, s) = \mathbb{E}\left[y(T) \mid  \boldsymbol{X} = \boldsymbol{x}, A = 1, S = s, \Gamma = 0\right]$, and 
\begin{equation*}{\textstyle
R(a; \Gamma = 1) = \mathbb{E}\left[\frac{\Gamma}{\kappa} \int_{s\in \mathcal{S}} \mu( \boldsymbol{X}, 1, s) f(s \mid  \boldsymbol{X}, A = a, \Gamma = 1) ds\right].}
\end{equation*}
\end{proposition}

Finally, let $C$ denote time-to-right-censoring, such as the time until a participant receives an out-of-study vaccination or is lost to follow-up. Instead of observing the true event times $T$ in the historical dataset $\mathcal{D}_h$, we observe the censored times $\Tilde{T}_i = \min(T_i, C_i)$ and the event indicator $\Delta_i = I\{T_i \leq C_i\}$. To still identify the target parameter, we make the following ignorable censoring assumption in the historical trial $(\Gamma = 0)$. 

\begin{assumption}\label{ass: standard right-censoring assumption}
{\rm (i. ignorable censoring)} $C \perp T \mid  \boldsymbol{X}, A, S$; {\rm (ii. finite horizon)} The transformation $y(\cdot)$ admits a maximal horizon $0 < H < \infty$ such that $y(t) = y(H)$ for all $t > H$; {\rm (iii. positivity of censoring)} $P(C < H \mid \boldsymbol{X}, A) < 1$.
\end{assumption}

Theorem \ref{thm:identification-censored-data} is our key identification result for the target parameter $R(a; \Gamma = 1),~ a\in\{1, 1'\}$, in the presence of covariate-dependent right-censoring.

\begin{theorem}[Identification for Task I]\label{thm:identification-censored-data}
Under Assumptions \ref{ass: standard assumpion in Dh}--\ref{ass: standard right-censoring assumption}, for $a \in \{1, 1'\}$, we have
\begin{equation*}{\textstyle
R(a; \Gamma = 1) = \mathbb{E}\left[\frac{(1 - \Gamma) I\{A = 1\}}{\kappa} \frac{f(\Gamma = 1 \mid  \boldsymbol{X})}{f(A = 1, \Gamma = 0 \mid  \mathbf{X})} \frac{f(S \mid \mathbf{X}, A = a, \Gamma = 1)}{f(S \mid \mathbf{X}, A = 1, \Gamma = 0)} \frac{\Delta y(\Tilde{T})}{G^{C}(\Tilde{T} \mid \mathbf{X}, A, S)}\right],}
\end{equation*}
where $G^{C}(t \mid \mathbf{x}, a, s) = P(C > t \mid \mathbf{X} = \mathbf{x}, A = a, S = s, \Gamma = 0)$ is the conditional survival function of censoring, 
\begin{equation*}{\textstyle
R(a; \Gamma = 1) = \mathbb{E}\left[\frac{\Gamma I\{A = a\}}{\kappa f(A = a \mid \mathbf{X}, \Gamma = 1)} \mu(\mathbf{X}, A, S)\right],}
\end{equation*}
and 
\begin{equation*}{\textstyle
R(a; \Gamma = 1) = \mathbb{E}\left[\frac{\Gamma}{\kappa} \int_{s\in \mathcal{S}} \mu( \mathbf{X}, a, s) f(s \mid  \mathbf{X}, A = a, \Gamma = 1) ds\right].}
\end{equation*}
\end{theorem}

\subsection{Task II}
\label{subsec: identification assumption Task II}

Recall in Task II, the investigational vaccine is matched to a circulating strain different from the strain that dominated during the historical VE trial, and the immunogenicity trial measures the immune marker against this new strain for both the approved and investigational vaccines. In this case, the historical dataset $\mathcal{D}_{h} = \{(\boldsymbol{X}_i, A_i, S_i, T_i): i = 1, \ldots, n_h\}$ remains the same as in Task I, while the immunobridging dataset now becomes $\mathcal{D}_{b} = \{(\boldsymbol{X}_i, A_i, S'_{i}) : i = n_h + 1, \ldots, n_h + n_b\}$, where  $S'_i$ is the same type of post-vaccination immune marker as $S_i$ in $\mathcal{D}_h$, but measured against a different pathogen strain.

To still identify the target parameter $R'(A = 1'; \Gamma = 1),$ we consider the following assumptions:

\begin{assumption}\label{ass: standard assumpion in Db for Task II}
{\rm (i. Consistency)} $S' = S'(A)$, $T' = T'(A, S')$ and $T' = T'(A) = T'(A, S'(A))$; {\rm (ii. Randomization)} $A \perp T'(a,s) \mid \boldsymbol{X}$ for $a \in \{1, 1'\}$; {\rm (iii. Strong sequential ignorability)} $S' \perp T'(a,s) \mid \boldsymbol{X}$ for $a \in \{1, 1'\}$; {\rm (iv. Positivity)} $P(A = a \mid \boldsymbol{X} = \boldsymbol{x}) > 0$ and $f(s \mid \boldsymbol{x}, a) > 0$ for all $a \in \{1, 1'\}, \boldsymbol{x} \in \mathcal{X}, s \in \mathcal{S}$.
\end{assumption}

\begin{assumption}[Conditional variant-invariant model]\label{ass: conditional exchangeability task II} A conditional variant-invariant model is said to hold for (i) the approved vaccine ($A = 1$) in the historical VE trial setting ($\Gamma = 0$) and (ii) the investigational, variant-matched vaccine ($A = 1'$) in the immunobridging-trial setting ($\Gamma = 1$) if
\begin{equation}
\label{eq: conditional variant-invariant model}
    \frac{\mathbb{E}\left[y(T(A = 1, S = s)) \mid \mathbf{X} = \boldsymbol{x}, \Gamma = 0\right]}{R(A = 0; \mathbf{X} = \boldsymbol{x}, \Gamma = 0)} = \frac{\mathbb{E}\left[y(T'(A = 1', S' = s)) \mid  \mathbf{X} = \boldsymbol{x}, \Gamma = 1\right]}{R'(A = 0; \mathbf{X} = \boldsymbol{x}, \Gamma = 1)},
\end{equation}
for $\boldsymbol{x} \in \mathcal{X}, s \in \mathcal{S}$ where $f(s \mid \boldsymbol{x}, a) > 0$, $a \in \{1, 1'\}$.
\end{assumption}

\noindent 
Assumption \ref{ass: standard assumpion in Db for Task II} is analogous to Assumption \ref{ass: standard assumpion in Dh}. We assume Assumption \ref{ass: standard assumpion in Db for Task II} for the immunobridging trial. Assumption \ref{ass: conditional exchangeability task II} is a ``conditional-on-$\boldsymbol{X}$" formalization of the ``variant-invariant CoP model" concept (termed by Jerry Sadoff in a different setting) in the immune correlates literature \citep{benkeser2023immune,luedtke2025immune}. 

Rearranging Equation \eqref{eq: conditional variant-invariant model} yields the following:
\begin{equation*}
    \mathbb{E}\left[y(T'(A = 1', S' = s)) \mid  \mathbf{X}, \Gamma = 1\right] = \underbrace{\frac{R'(A = 0; \mathbf{X}, \Gamma = 1)}{R(A = 0; \mathbf{X}, \Gamma = 0)}}_{\substack{\text{Relative} \\ \text{transmissibility} \\ \text{factor} }}\cdot\mathbb{E}\left[y(T(A = 1, S = s)) \mid \mathbf{X}, \Gamma = 0\right],
\end{equation*}
which links the controlled risk in the immunobridging trial context under the joint intervention $(A = 1', S' = s)$ against the currently circulating strain to the controlled risk in the historical trial context under the joint intervention $(A = 1, S = s)$ against the then-circulating strain, multiplied by a \emph{relative transmissibility factor}. This factor is defined as the ratio of placebo risks across the two trial contexts---one against the then-circulating strain and the other against the currently circulating strain---and is intended to partially account for differences in: (1) the amount of virus circulating and exposing individuals in the two settings; and (2) transmissibility, viral fitness, and viral load across strains. The quantity $R(A = 0; \mathbf{X}, \Gamma = 0)$ can be unbiasedly estimated from $\mathcal{D}_h$ under suitable identification assumptions, whereas $R'(A = 0; \mathbf{X}, \Gamma = 1)$ is not identifiable without external surveillance data, since $\mathcal{D}_b$ does not collect clinical endpoint data. We therefore propose anchoring the relative transmissibility factor at $1$ and varying it as a sensitivity parameter.

To further identify the target parameter $R'(A = 1; \Gamma = 1)$, we consider the following assumption:

\begin{assumption}[No controlled direct effects comparing investigational and approved vaccines in the immunobridging trial]\label{ass: same joint intervention task II} The no controlled direct effects assumption is said to hold for the approved vaccine $A = 1$, the investigational vaccine $A = 1'$, and the immune marker $S'$ in the immunobridging trial if the following holds:
\begin{equation}
\mathbb{E}\left[y(T'(A = 1', S' = s)) \mid  \mathbf{X} = \boldsymbol{x}, \Gamma = 1\right] = \mathbb{E}\left[y(T'(A = 1, S' = s)) \mid  \mathbf{X} = \boldsymbol{x}, \Gamma = 1\right],
\end{equation}
for $\boldsymbol{x} \in \mathcal{X}, s \in \mathcal{S}$ where $f(s \mid \boldsymbol{x}, a) > 0$, $a \in \{1, 1'\}$.
\end{assumption}

Together, Assumptions \ref{ass: conditional exchangeability task II} and \ref{ass: same joint intervention task II} help connect the potential risk with respect to $T'$ under the joint intervention $(A = 1, S' = s)$ in the immunobridging-trial context to the risk with respect to $T$ under the joint intervention $(A = 1, S = s)$ in the historical-trial context, and the latter can be identified from the historical VE trial data. 

Proposition \ref{prop: identification Task II} formalizes the identification results for Task II. The identification results are essentially the same as those obtained in Proposition~\ref{prop:identification-complete-data}. The explicit formulas are deferred to Supplemental Materials Section~\ref{proof-prop:identification-complete-data}.

\begin{proposition}[Identification for Task II]\label{prop: identification Task II}
    Under Assumptions \ref{ass: standard assumpion in Dh}, \ref{ass: overlap}, \ref{ass: standard right-censoring assumption}, \ref{ass: standard assumpion in Db for Task II}, and \ref{ass: conditional exchangeability task II}, the target parameter $R'(A = 1'; \Gamma = 1)$ for Task II can be identified from the historical dataset $\mathcal{D}_{h} = \{(\boldsymbol{X}_i, A_i, S_i, T_i): i = 1, \ldots, n_h\}$ and the immunobridging dataset $\mathcal{D}_{b} = \{(\boldsymbol{X}_i, A_i, S'_{i}) : i = n_h + 1, \ldots, n_h + n_b\}$. If one further makes Assumption \ref{ass: same joint intervention task II}, then $R'(A = 1; \Gamma = 1)$ and hence the relative VE comparing $R'(A = 1'; \Gamma = 1)$ and $R'(A = 1; \Gamma = 1)$ can be identified.
\end{proposition}

Analogous to the Assumptions \ref{ass: conditional exchangeability task I} and \ref{ass: same joint intervention A = 1 & 1' Task I}, Assumptions \ref{ass: conditional exchangeability task II} and \ref{ass: same joint intervention task II} are not testable; therefore, researchers may want to use them as ``leading cases, not truths" \citep{tukey1986sunset}, and relax them using methods discussed in Remark \ref{remark: relax Assumptions 3 & 4}.

%% file: 4.Estimation-Inference.tex
\vspace{-10pt}
\section{Estimation and Inference for Tasks I and II}
\label{sec: estimation and inference}
\subsection{Uncensored outcome}
\label{subsec: estimation and inference uncensored outcome}
Proposition \ref{prop:EIF-complete-data} derives the efficient influence function (EIF) of the target parameter $R(a; \Gamma = 1)$ for an uncensored outcome in Proposition \ref{prop:identification-complete-data}. Proposition \ref{prop:EIF-complete-data} is the same as that derived in \citet{athey2025surrogate} and \citet{gilbert2024surrogateendpointbasedprovisional}, although the problem setup is slightly different. We nevertheless include it here for completeness, and note that the multiple robustness property and debiased machine learning approach were not explicitly stated in \citet{athey2025surrogate} or \citet{gilbert2024surrogateendpointbasedprovisional}.

\begin{proposition}\label{prop:EIF-complete-data}
The parameter $R(a; \Gamma = 1)$ is pathwise differentiable in a nonparametric model for complete data, with efficient influence function $\phi^\ast_a = \phi_a - R(a; \Gamma = 1)$, where $\phi_a$ equals
{\small
\begin{align*}
& \frac{(1 - \Gamma) I\{A = 1\}}{\kappa} \frac{f(\Gamma = 1 \mid \boldsymbol{X})}{f(A = 1, \Gamma = 0 \mid \boldsymbol{X})} \frac{f(S \mid \boldsymbol{X}, A = a, \Gamma = 1)}{f(S \mid \boldsymbol{X}, A = 1, \Gamma = 0)} \left\{y(T) - \mu(\boldsymbol{X}, a, S)\right\} \\
& + \frac{\Gamma I\{A = a\}}{\kappa f(A = a \mid \boldsymbol{X}, \Gamma = 1)} \left\{\mu(\boldsymbol{X}, a, S) - \int_{s\in \mathcal{S}} \mu(\boldsymbol{X}, a, s) f(s \mid \boldsymbol{X}, A = a, \Gamma = 1) ds\right\} \\
& + \frac{\Gamma}{\kappa} \int_{s\in \mathcal{S}} \mu(\boldsymbol{X}, a, s) f(s \mid  \boldsymbol{X}, A = a, \Gamma = 1) ds.
\end{align*}}
\end{proposition}

While we assume a nonparametric model, the treatment assignment probabilities are known by design in our setup, and note that knowledge of treatment assignment should not change the EIF \citep{hahn1998role}. We can thus use the EIF to construct a standard semiparametrically efficient estimator
$\hat{R}(a; \Gamma = 1) = \frac{1}{n} \sum_{i=1}^{n} \hat{\phi}_a,$
where $n = n_h + n_b$, and $\hat{\phi}_a$ is defined as the function $\phi_a$ in Proposition \ref{prop:EIF-complete-data} with the nuisance functions estimated via parametric or semiparametric models. 

The multiple robustness property of $\hat{R}(a; \Gamma = 1)$ is described in Proposition \ref{prop:MR-complete-data}.

\begin{proposition}\label{prop:MR-complete-data}
Under standard regularity conditions, $\hat{R}(a; \Gamma = 1)$ is a consistent and asymptotically normal estimator of $R(a; \Gamma = 1)$ under the union of the following three models:
\begin{description}
    \item[$\mathcal{M}_a$]: models for $\mu(\boldsymbol{x}, a, s)$ and $f(s \mid \boldsymbol{x}, A = a, \Gamma = 1)$ are correctly specified,

    \item[$\mathcal{M}_b$]: models for $\mu(\boldsymbol{x}, a, s)$ and $f(A = a \mid \boldsymbol{x}, \Gamma = 1)$ are correctly specified,

    \item[$\mathcal{M}_c$]: models for $f(A = a \mid \boldsymbol{x}, \Gamma = 1)$, $f(\Gamma = 1 \mid \boldsymbol{x})$, $f(A = 1, \Gamma = 0 \mid \boldsymbol{x})$, $f(s \mid \boldsymbol{x}, A = a, \Gamma = 1)$ and $f(s \mid \boldsymbol{x}, A = 1, \Gamma = 0)$ are correctly specified.
\end{description}
\end{proposition}

When parametric and semiparametric models are used to estimate nuisance functions, standard M-estimation theory can be used to obtain a nonparametric sandwich estimator for the standard error of $\hat{R}(a;\Gamma=1)$. The resulting sandwich variance estimator is robust to partial model misspecification. As an alternative, inference can be based on the nonparametric bootstrap with nuisance functions estimated via parametric or semiparametric models \citep{cheng2010bootstrap}.

When flexible machine learning methods are used to estimate the nuisance functions, we propose the following debiased machine learning (cross-fitted one-step) estimator \citep{chernozhukov2018double}
\begin{equation*}
    \hat{R}_{\textrm{DML}}(a; \Gamma = 1) = \frac{1}{n} \sum_{k=1}^{K} \sum_{i \in \mathcal{I}_{n,k}} \hat{\phi}_{a},
\end{equation*}
where we define a $K$-fold random partition $(\mathcal{I}_{n,k})_{k = 1, \ldots, K}$ of $\{1, \ldots, n\}$ for some fixed $K$. For each $k \in \{1, 2, \ldots, K\}$, the observations outside fold $k$ constitute the training set, and for $i \in \mathcal{I}_{n,k}$, the nuisance functions in $\hat{\phi}_{a}$ are estimated via machine learning methods fitted on the corresponding training set. We show the following asymptotic linearity to achieve valid inference.
\begin{proposition}\label{prop: asymptotic linearity complete data}
Under standard conditions, we have
\begin{equation*}
    \hat{R}_{\textrm{DML}}(a; \Gamma = 1) = R(a; \Gamma = 1) + \mathbb{P}_{n}\phi^{\ast}_{a} + o_{P}\left(n^{-1/2}\right),
\end{equation*}
which implies that $n^{1/2} \left\{\hat{R}_{\textrm{DML}}(a; \Gamma = 1) - R(a; \Gamma = 1)\right\}$ converges in distribution to a Gaussian random variable with mean $0$ and variance $\mathbb{E}\left[(\phi^{\ast}_{a})^2\right]$.
\end{proposition}

\subsection{Pointwise and uniform inference for a censored outcome}
\label{subsec: estimation and inference censored outcome}
By setting $y(T) = I\{T \leq t\}$ in Theorem \ref{thm:identification-censored-data}, the target parameter reduces to the following cumulative incidence curve:
\begin{equation*}
    t \mapsto R(a, t; \Gamma = 1) = \mathbb{E}_{\mathcal{P}_b}\left[P\{T(a) \leq t \mid \boldsymbol{X}, \Gamma = 1\}\right],
\end{equation*}
for $a \in \{1, 1'\}$ and $t \in [0, H]$ for some positive and finite $H$. We next study both pointwise and uniform inference for $R(a, t; \Gamma = 1)$ in the presence of covariate-dependent right-censoring.

Theorem \ref{thm:EIF-censored-data} is our main result and establishes the EIF for the target parameter $R(a, t; \Gamma = 1)$.

\begin{theorem}\label{thm:EIF-censored-data}
The parameter $R(a, t; \Gamma = 1)$ is pathwise differentiable in a nonparametric model for censored data, with efficient influence function $\phi^{C\ast}_{a,t} = \phi^{C}_{a,t} - R(a, t; \Gamma = 1)$, where $\phi^{C}_{a,t}$ equals
{\small
\begin{align*}
& 1 - \frac{(1 - \Gamma) I\{A = 1\}}{\kappa} \frac{f(\Gamma = 1 \mid \boldsymbol{X})}{f(A = 1, \Gamma = 0 \mid \boldsymbol{X})} \frac{f(S \mid \boldsymbol{X}, A = a, \Gamma = 1)}{f(S \mid \boldsymbol{X}, A = 1, \Gamma = 0)} \left\{\frac{-I\{\Tilde{T} \leq t, \Delta = 1\} G^{T}(t \mid \boldsymbol{X}, A, S)}{G^{T}(\Tilde{T} \mid \boldsymbol{X}, A, S) G^{C}(\Tilde{T} \mid \boldsymbol{X}, A, S)} \right. \\
& \left.\quad + \int_{0}^{t \wedge \Tilde{T}} \frac{G^{T}(t \mid \boldsymbol{X}, A, S) \Lambda(du \mid \boldsymbol{X}, A, S)}{G^{T}(u \mid \boldsymbol{X}, A, S) G^{C}(u \mid \boldsymbol{X}, A, S)}\right\} \\
& - \frac{\Gamma I\{A = a\}}{\kappa f(A = a \mid \boldsymbol{X}, \Gamma = 1)} \left\{G^{T}(t \mid \boldsymbol{X}, a, S) - \int_{s\in \mathcal{S}} G^{T}(t \mid \boldsymbol{X}, a, s) f(s \mid \boldsymbol{X}, A = a, \Gamma = 1) ds\right\} \\
& - \frac{\Gamma}{\kappa} \int_{s\in \mathcal{S}} G^{T}(t \mid \boldsymbol{X}, a, s) f(s \mid  \boldsymbol{X}, A = a, \Gamma = 1) ds,
\end{align*}}
where $G^{T}(t \mid \boldsymbol{x}, a, s) = P\left[T > t \mid \boldsymbol{X} = \boldsymbol{x}, A = a, S = s, \Gamma = 0\right]$ is the conditional survival function, $\Lambda(t \mid \boldsymbol{x}, a, s)$ is the conditional cumulative hazard function of $T$, and $G^C$ is defined in Theorem \ref{thm:identification-censored-data}.
\end{theorem}

Similarly, we can use the EIF in Theorem \ref{thm:EIF-censored-data} to construct a semiparametrically efficient estimator
\begin{equation*}
    \hat{R}^{C}(a, t; \Gamma = 1) = \frac{1}{n} \sum_{i=1}^{n} \hat{\phi}_{a,t}^{C},
\end{equation*}
where $\hat{\phi}_{a,t}^{C}$ is defined as the function $\phi^{C}_{a,t}$ where the nuisance functions are estimated via (semi)parametric models. Theorem \ref{thm:MR-censored-data} establishes the multiple robustness property of $\hat{R}^{C}(a, t; \Gamma = 1)$.

\begin{theorem}\label{thm:MR-censored-data}
Under standard regularity conditions, $\hat{R}^{C}(a, t; \Gamma = 1)$ is a consistent and asymptotically normal estimator of $R(a, t; \Gamma = 1)$ under the union of the following three models:
\begin{description}
    \item[$\mathcal{M}'_a$]: models for $G^{T}(t \mid \boldsymbol{x}, a, s)$ and $f(s \mid \boldsymbol{x}, A = a, \Gamma = 1)$ are correctly specified,

    \item[$\mathcal{M}'_b$]: models for $G^{T}(t \mid \boldsymbol{x}, a, s)$ and $f(A = a \mid \boldsymbol{x}, \Gamma = 1)$ are correctly specified,

    \item[$\mathcal{M}'_c$]: models for $f(A = a \mid \boldsymbol{x}, \Gamma = 1)$, $f(\Gamma = 1 \mid \boldsymbol{x})$, $f(A = 1, \Gamma = 0 \mid \boldsymbol{x})$, $f(s \mid \boldsymbol{x}, A = a, \Gamma = 1)$, $f(s \mid \boldsymbol{x}, A = 1, \Gamma = 0)$ and $G^{C}(t \mid \boldsymbol{x}, a, s)$ are correctly specified.
\end{description}
\end{theorem}

As in Proposition~\ref{prop:MR-complete-data}, the standard error of $\hat{R}^{C}(a, t; \Gamma = 1)$ can be obtained using a standard sandwich estimator or nonparametric bootstrap when nuisance functions are estimated via parametric or semiparametric models. When using flexible machine learning methods, we propose the debiased machine learning estimator
\begin{equation*}
    \hat{R}^{C}_{\textrm{DML}}(a, t; \Gamma = 1) = \frac{1}{n} \sum_{k=1}^{K} \sum_{i \in \mathcal{I}_{n,k}} \hat{\phi}_{a,t}^{C}.
\end{equation*}

The following asymptotic linearity result serves as the basis of conducting pointwise and uniform inference.

\begin{proposition}\label{prop: asymptotic linearity censored data}
Under standard conditions given in the Supplemental Materials Section~\ref{sec:SM proof asymptotic linearity censored data}, we have
\begin{equation*}
    \hat{R}^{C}_{\textrm{DML}}(a, t; \Gamma = 1) = R(a, t; \Gamma = 1) + \mathbb{P}_{n}\phi^{C\ast}_{a,t} + o_{P}\left(n^{-1/2}\right).
\end{equation*}
\end{proposition}

Proposition~\ref{prop: asymptotic linearity censored data} implies that $n^{1/2} \left\{\hat{R}^{C}_{\textrm{DML}}(a, t; \Gamma = 1) - R(a, t; \Gamma = 1)\right\}$ converges in distribution to a Gaussian random variable with mean $0$ and variance $\mathbb{E}\left[(\phi^{C\ast}_{a,t})^2\right]$. 

If additional standard uniformity conditions are assumed \citep{westling2024inference}, we also have
\begin{equation*}
    \sup_{u \in [0,t]} \left|\hat{R}^{C}_{\textrm{DML}}(a, u; \Gamma = 1) - R(a, u; \Gamma = 1) - \mathbb{P}_{n}\phi^{C\ast}_{a,u}\right| = o_{P}\left(n^{-1/2}\right),
\end{equation*}
which then implies that \[
\left\{n^{1/2} \left[\hat{R}^{C}_{\textrm{DML}}(a, u; \Gamma = 1) - R(a, u; \Gamma = 1)\right] : u \in [0, t]\right\}
\] converges weakly to a tight mean-zero Gaussian process with covariance function $(u,v) \mapsto \mathbb{E}\left[\phi^{C\ast}_{a,u} \phi^{C\ast}_{a,v}\right]$.

For a given dataset, the estimated cumulative incidence evaluated at a set of time points $t$ may not be monotone non-decreasing in $t$. To enforce monotonicity, one can project the estimated cumulative incidence curve onto the space of non-decreasing functions using isotonic regression \citep{westling2020correcting,westling2024inference}. \citet{westling2020correcting} demonstrate that the resulting ``monotone-corrected” cumulative incidence function performs at least as well as the uncorrected curve in finite samples, and that the two curves---the original and the monotone-corrected---are asymptotically equivalent. \citet{westling2024inference} provide a general procedure to conduct asymptotically valid pointwise and uniform inference for survival curves.

Finally, a natural estimator of the relative vaccine efficacy, defined as $1 - R(A = 1', t; \Gamma = 1)/R(A = 1, t; \Gamma = 1)$, can be obtained using $1 - \widehat{R}(A = 1', t; \Gamma = 1)/\widehat{R}(A = 1, t; \Gamma = 1)$. Proposition \ref{prop: relative risk} states the property of the estimator and derives the new EIF via the delta method.

\begin{proposition}\label{prop: relative risk}
    Under the same standard conditions as assumed in Proposition~\ref{prop: asymptotic linearity censored data}, the estimator $1-\widehat{R}(A=1', t; \Gamma = 1)/\widehat{R}(A=1, t; \Gamma = 1)$ is consistent and asymptotically normal for the relative vaccine efficacy. The variance of $\log\{1 - \widehat{R}(A=1', t; \Gamma = 1)/\widehat{R}(A=1, t; \Gamma = 1)\}$ equals
    \begin{equation*}
        \mathbb{E}\left\{\frac{1}{R(1, t; \Gamma = 1) - R(1', t; \Gamma = 1)}\left(\phi^{C\ast}_{1',t} - \frac{R(1', t; \Gamma = 1)}{R(1, t; \Gamma = 1)}\phi^{C\ast}_{1,t}\right)\right\}^2
    \end{equation*}
    and can be estimated via the usual plug-in estimator.
\end{proposition}

Similarly, results in Theorem \ref{thm:EIF-censored-data} and Proposition \ref{prop: asymptotic linearity censored data}, combined with the delta method, can be used to perform pointwise and uniform inference for a generic contrast of the form 
$g\{R(A=1, t; \Gamma = 1), R(A=1', t; \Gamma = 1)\}$
for any differentiable function $g$.

%% file: 5.Competing-risk.tex
\section{Extension to multiple serotype-specific endpoints}
\label{sec: extension to competing risk}
\subsection{Target parameter}
\label{subsec: competing risk target parameter}
For some pathogens, disease outcomes can be classified into distinct serotypes based on variations in bacterial or viral surface antigens. For example, Group B Streptococcus (GBS) comprises up to ten Capsular Polysaccharide serotypes globally, and most vaccine candidates under development target several of the most prevalent ones—such as serotypes Ia, Ib, II, III, and V \citep{heath2016status}. The dengue virus has four serotypes (DENV-1 through DENV-4) \citep{deng2020review}, and vaccine development efforts aim to achieve protection against all four. Similarly, influenza viruses exhibit substantial antigenic diversity and are classified into multiple types and subtypes. Seasonal influenza vaccines in the United States are formulated to protect against three major circulating strains, typically including one influenza A(H1N1) virus, one influenza A(H3N2) virus, and one influenza B virus from the Victoria lineage (see, e.g., \citet{jackson2017influenza}).



In vaccine efficacy trials, a commonly used endpoint is the time to the first occurrence of a serotype-specific event. Let $j = 1, \dots, J$ denote the $J$ circulating pathogen strains. Our target parameter is the following cause-$j$-specific cumulative incidence function:
\begin{equation}\label{eq: conditional cause-specific risk}
    R^j(A = 1', t; \Gamma = 1) = \mathbb{E}_{\mathcal{P}_b}\left[P\{T(1') \leq t, \Delta(1') = j \mid \boldsymbol{X}, \Gamma = 1\}\right],
\end{equation}
in a hypothetical trial evaluating a candidate vaccine regimen $A = 1'$ in the target population, where $\{T(1'), \Delta(1')\}$ denotes the potential time to the first failure event under the intervention $A = 1'$, and $\Delta$ is the failure event indicator with $0$ for censoring. The all-cause cumulative incidence function in the target trial is then defined by $R(A = 1', t;\Gamma = 1) = \sum_{j=1}^J R^j(A = 1', t; \Gamma = 1)$. Let $C$ denote time to right-censoring, so the observed time is $\Tilde{T} = \min(T, C)$.

For each cause-specific time-to-event outcome, we consider the corresponding cause-specific mediator $S^j$. In the GBS example, $S^j$ denotes the IgG binding antibody level in cord blood against the $j$-th GBS serotype. In the dengue vaccine example, \(S^j\) represents the post-vaccination binding or neutralizing antibody level against the $j$-th dengue virus serotype. We write $\boldsymbol{S} = (S^1, \ldots, S^J)$. With a cause-specific time-to-event endpoint, the historical dataset is $\mathcal{D}_{h} = \{(\boldsymbol{X}_i, A_i, \boldsymbol{S}_i, \tilde{T}_i, \Delta_i) : i = 1, \ldots, n_h\},
$ and the immunobridging dataset is
$\mathcal{D}_{b} = \{(\boldsymbol{X}_i, A_i, \boldsymbol{S}_i) : i = n_h + 1, \ldots, n_h + n_b\}.$

\subsection{Identification and inference for Task III}
\label{subsec: competing risk identification and inference}
Assumption \ref{ass: standard assumpion in Dh competing risks} is an analogue of Assumption \ref{ass: standard assumpion in Dh} under a competing risks framework. We assume Assumption \ref{ass: standard assumpion in Dh competing risks} holds for the historical trial $\Gamma = 0$ and the hypothetical trial $\Gamma = 1$.

\begin{assumption}\label{ass: standard assumpion in Dh competing risks}
{\rm (i. Consistency)} $\boldsymbol{S} = \boldsymbol{S}(A)$, $T = T(A, \boldsymbol{S})$, $\Delta = \Delta(A, \boldsymbol{S})$, $T = T(A) = T(A, \boldsymbol{S}(A))$ and $\Delta = \Delta(A) = \Delta(A, \boldsymbol{S}(A))$; {\rm (ii. Randomization)} $A \perp \{T(a,\boldsymbol{s}), \Delta(a,\boldsymbol{s})\} \mid \boldsymbol{X}$ for $a \in \{0, 1, 1'\}$; {\rm (iii. Strong sequential ignorability)} $\boldsymbol{S} \perp \{T(a,\boldsymbol{s}), \Delta(a,\boldsymbol{s})\} \mid \boldsymbol{X}$ for $a \in \{0, 1, 1'\}$; {\rm (iv. Positivity)} $P(A = a \mid \boldsymbol{X}) > 0$ and $f(\boldsymbol{s} \mid \boldsymbol{x}, a) > 0$.
\end{assumption}

Assumptions \ref{ass: conditional exchangeability competing risks} and \ref{ass: same joint intervention competing risks} are analogous to Assumptions \ref{ass: conditional exchangeability task I} and \ref{ass: same joint intervention A = 1 & 1' Task I}, but stated with respect to the cause-$j$-specific endpoint.

\begin{assumption}[Conditional exchangeability of the approved vaccine for the cause-$j$-specific endpoint in two trial settings]\label{ass: conditional exchangeability competing risks} Conditional exchangeability across historical and hypothetical trials is said to hold if 
\[
P\left(T(A = 1, \boldsymbol{S} = \boldsymbol{s}) \leq t, \Delta = j \mid \mathbf{X} = \boldsymbol{x}, \Gamma = 0\right) = P\left(T(A = 1, \boldsymbol{S} = \boldsymbol{s}) \leq t, \Delta = j \mid  \mathbf{X} = \boldsymbol{x}, \Gamma = 1\right),
\]
for $\boldsymbol{x} \in \mathcal{X}, \boldsymbol{s} \in \mathcal{S}$ where $f(\boldsymbol{s} \mid \boldsymbol{x}, 1) > 0$.
\end{assumption}

\begin{assumption}[No controlled direct effects for cause-$j$-specific endpoint comparing investigational and approved vaccines in the immunobridging trial]\label{ass: same joint intervention competing risks} The no controlled direct effects assumption is said to hold for the approved vaccine $A = 1$, the investigational vaccine $A = 1'$, the immune marker $\boldsymbol{S}$, and the cause-$j$-specific endpoint if the following holds:
\[
P\left(T(A = 1, \boldsymbol{S} = \boldsymbol{s}) \leq t, \Delta = j \mid \mathbf{X} = \boldsymbol{x}, \Gamma = 1\right) = P\left(T(A = 1', \boldsymbol{S} = \boldsymbol{s}) \leq t, \Delta = j \mid  \mathbf{X} = \boldsymbol{x}, \Gamma = 1\right),
\]
for $\boldsymbol{x} \in \mathcal{X}, \boldsymbol{s} \in \mathcal{S}$ where $f(\boldsymbol{s} \mid \boldsymbol{x}, a) > 0$, $a \in \{1, 1'\}$.
\end{assumption}


Theorem \ref{thm:identification-complete-data competing risk} identifies the cause-$j$-specific risk in the absence of censoring, and Theorem \ref{thm:EIF-complete-data competing risk} derives its EIF. Theorems \ref{thm:identification-complete-data competing risk} and \ref{thm:EIF-complete-data competing risk} are competing risks extensions of Propositions \ref{prop:identification-complete-data} and \ref{prop:EIF-complete-data}.

\begin{theorem}\label{thm:identification-complete-data competing risk}
Under Assumptions~\ref{ass: overlap}, \ref{ass: standard assumpion in Dh competing risks}, \ref{ass: conditional exchangeability competing risks}, and \ref{ass: same joint intervention competing risks}, for $a \in \{1, 1'\}$, we have
{\small
\begin{equation*}
R^j(a, t; \Gamma = 1) = \mathbb{E}\left[\frac{(1 - \Gamma) I\{A = 1\}}{\kappa} \frac{f(\Gamma = 1 \mid \boldsymbol{X})}{f(A = 1, \Gamma = 0 \mid \boldsymbol{X})} \frac{f(\boldsymbol{S} \mid \boldsymbol{X}, A = a, \Gamma = 1)}{f(\boldsymbol{S} \mid \boldsymbol{X}, A = 1, \Gamma = 0)} I\{T \leq t, \Delta = j\}\right],
\end{equation*}}
{\small
\begin{equation*}
R^j(a, t; \Gamma = 1) = \mathbb{E}\left[\frac{\Gamma I\{A = a\}}{\kappa f(A = a \mid \boldsymbol{X}, \Gamma = 1)} F^{T,j}(t \mid \boldsymbol{X}, 1, \boldsymbol{S})\right],
\end{equation*}
}
where $F^{T,j}(t \mid \boldsymbol{x}, a, \boldsymbol{s}) = P\left[T \leq t, \Delta = j \mid \boldsymbol{X} = \boldsymbol{x}, A = a, \boldsymbol{S} = \boldsymbol{s}, \Gamma = 0\right]$, and 
\begin{equation*}
R^j(a, t; \Gamma = 1) = \mathbb{E}\left[\frac{\Gamma}{\kappa} \int_{\boldsymbol{s} \in \boldsymbol{\mathcal{S}}} F^{T,j}(t \mid \boldsymbol{X}, 1, \boldsymbol{s}) f(\boldsymbol{s} \mid \boldsymbol{X}, A = a, \Gamma = 1) d\boldsymbol{s}\right].
\end{equation*}
\end{theorem}

\begin{theorem}
\label{thm:EIF-complete-data competing risk}
The parameter $R^j(a, t; \Gamma = 1)$ is pathwise differentiable in a nonparametric model for complete competing risks data, with efficient influence function $\phi^{j\ast}_a = \phi^j_a - R^j(a, t; \Gamma = 1)$, where $\phi^j_a$ equals
{\small
\begin{align*}
& \frac{(1 - \Gamma) I\{A = 1\}}{\kappa} \frac{f(\Gamma = 1 \mid \boldsymbol{X})}{f(A = 1, \Gamma = 0 \mid \boldsymbol{X})} \frac{f(\boldsymbol{S} \mid \boldsymbol{X}, A = a, \Gamma = 1)}{f(\boldsymbol{S} \mid \boldsymbol{X}, A = 1, \Gamma = 0)} \left\{I\{T \leq t, \Delta = j\} - F^{T,j}(t \mid \boldsymbol{X}, 1, \boldsymbol{S}) \right\} \\
& + \frac{\Gamma I\{A = a\}}{\kappa f(A = a \mid \boldsymbol{X}, \Gamma = 1)} \left\{F^{T,j}(t \mid \boldsymbol{X}, 1, \boldsymbol{S}) - \int_{\boldsymbol{s} \in \boldsymbol{\mathcal{S}}} F^{T,j}(t \mid \boldsymbol{X}, 1, \boldsymbol{s}) f(\boldsymbol{s} \mid \boldsymbol{X}, A, \Gamma = 1) d\boldsymbol{s}\right\} \\
& + \frac{\Gamma}{\kappa} \int_{\boldsymbol{s} \in \boldsymbol{\mathcal{S}}} F^{T,j}(t \mid \boldsymbol{X}, 1, \boldsymbol{s}) f(\boldsymbol{s} \mid  \boldsymbol{X}, A = a, \Gamma = 1) d\boldsymbol{s}.
\end{align*}}
\end{theorem}

Finally, Theorems \ref{thm:identification-censored-data competing risks} and \ref{thm:EIF-censored-data competing risks} are analogous to Theorems \ref{thm:identification-censored-data} and \ref{thm:EIF-censored-data} and derive the identification functional and the corresponding efficient influence function for cause-$j$-specific risk under the ignorable censoring assumption formalized in Assumption \ref{ass: standard right-censoring assumption competing risks}.

\begin{assumption}\label{ass: standard right-censoring assumption competing risks}
{\rm (i. ignorable censoring)} $C \perp \{T, \Delta\} \mid  \boldsymbol{X}, A, \boldsymbol{S}$; {\rm (ii. finite horizon)} The transformation $y(\cdot)$ admits a maximal horizon $0 < H < \infty$ such that $y(t) = y(H)$ for all $t > H$; {\rm (iii. positivity of censoring)} $P(C < H \mid \boldsymbol{X}, A, \boldsymbol{S}) < 1$.
\end{assumption}

\begin{theorem}\label{thm:identification-censored-data competing risks}
Under Assumptions~\ref{ass: overlap} and \ref{ass: standard assumpion in Dh competing risks}--\ref{ass: standard right-censoring assumption competing risks}, for $a \in \{1, 1'\}$, we have
{\small
\begin{equation*}
R^j(a, t; \Gamma = 1) = \mathbb{E}\left[\frac{(1 - \Gamma) I\{A = 1\}}{\kappa} \frac{f(\Gamma = 1 \mid \boldsymbol{X})}{f(A = 1, \Gamma = 0 \mid  \boldsymbol{X})} \frac{f(\boldsymbol{S} \mid \boldsymbol{X}, A = a, \Gamma = 1)}{f(\boldsymbol{S} \mid \boldsymbol{X}, A = 1, \Gamma = 0)} \frac{I\{\Tilde{T} \leq t, \Delta = j\}}{G^{C}(\Tilde{T} \mid \boldsymbol{X}, A, \boldsymbol{S})}\right],
\end{equation*}}
where $G^{C}(t \mid \boldsymbol{x}, a, \boldsymbol{s}) = P\left(C > t \mid \boldsymbol{X} = \boldsymbol{x}, A = a, \boldsymbol{S} = \boldsymbol{s}, \Gamma = 0\right)$,
{\small
\begin{equation*}
R^j(a, t; \Gamma = 1) = \mathbb{E}\left[\frac{\Gamma I\{A = a\}}{\kappa f(A = a \mid \boldsymbol{X}, \Gamma = 1)} F^{T,j}(t \mid \boldsymbol{X}, 1, \boldsymbol{S})\right],
\end{equation*}}
and 
{\small
\begin{equation*}
R^j(a, t; \Gamma = 1) = \mathbb{E}\left[\frac{\Gamma}{\kappa} \int_{\boldsymbol{s} \in \boldsymbol{\mathcal{S}}} F^{T,j}(t \mid \boldsymbol{X}, 1, \boldsymbol{s}) f(\boldsymbol{s} \mid \boldsymbol{X}, A = a, \Gamma = 1) d\boldsymbol{s}\right].
\end{equation*}}
\end{theorem}

\begin{theorem}\label{thm:EIF-censored-data competing risks}
The parameter $R^j(a, t; \Gamma = 1)$ is pathwise differentiable in a nonparametric model for censored competing risks data, with efficient influence function $\phi^{C,j\ast}_{a,t} = \phi^{C,j}_{a,t} - R^j(a, t; \Gamma = 1)$, where $\phi^{C,j}_{a,t}$ equals
{\scriptsize
\begin{align*}
& \frac{(1 - \Gamma) I\{A = 1\}}{\kappa} \frac{f(\Gamma = 1 \mid \boldsymbol{X})}{f(A = 1, \Gamma = 0 \mid \boldsymbol{X})} \frac{f(\boldsymbol{S} \mid \boldsymbol{X}, A = a, \Gamma = 1)}{f(\boldsymbol{S} \mid \boldsymbol{X}, A = 1, \Gamma = 0)} \sum_{k=1}^{J} \int \frac{I\{u \leq t\} h^{k}(t,u)}{G^{C}(\Tilde{T} \mid \boldsymbol{X}, A, \boldsymbol{S})} \left\{N^{k}(du) - I\{\Tilde{T} \geq u\}\Lambda^{k}(du \mid \boldsymbol{X}, A, \boldsymbol{S})\right\} \\
& + \frac{\Gamma I\{A = a\}}{\kappa f(A = a \mid \boldsymbol{X}, \Gamma = 1)} \left\{F^{T,j}(t \mid \boldsymbol{X}, 1, \boldsymbol{S}) - \int_{\boldsymbol{s} \in \boldsymbol{\mathcal{S}}} F^{T,j}(t \mid \boldsymbol{X}, 1, \boldsymbol{s}) f(\boldsymbol{s} \mid \boldsymbol{X}, A, \Gamma = 1) d\boldsymbol{s}\right\} \\
& + \frac{\Gamma}{\kappa} \int_{\boldsymbol{s} \in \boldsymbol{\mathcal{S}}} F^{T,j}(t \mid \boldsymbol{X}, 1, \boldsymbol{s}) f(\boldsymbol{s} \mid \boldsymbol{X}, A = a, \Gamma = 1) d\boldsymbol{s},
\end{align*}}
where $N^{k}(u) = I\{T \leq u, \Delta = k\}$, $\Lambda^k(u \mid \boldsymbol{x}, a, \boldsymbol{s})$ is the conditional cumulative hazard function of $\{T, \Delta = k\}$, $G^{T}(u \mid \boldsymbol{x}, a, \boldsymbol{s}) = \exp{\left\{-\sum_{k=1}^{J} \Lambda^k(u \mid \boldsymbol{x}, a, \boldsymbol{s})\right\}}$, and we define
{
\begin{equation*}
    h^{k}(t,u) = \begin{cases}
        1 - \frac{F^{T,k}(t \mid \boldsymbol{X}, A, \boldsymbol{S}) - F^{T,k}(u \mid \boldsymbol{X}, A, \boldsymbol{S})}{G^{T}(u \mid \boldsymbol{X}, A, \boldsymbol{S})}, & \text{when } k = j,\\
        -\frac{F^{T,k}(t \mid \boldsymbol{X}, A, \boldsymbol{S}) - F^{T,k}(u \mid \boldsymbol{X}, A, \boldsymbol{S})}{G^{T}(u \mid \boldsymbol{X}, A, \boldsymbol{S})}, & \text{when } k \neq j. 
    \end{cases}
\end{equation*}}
\end{theorem}

Equipped with the EIF, a debiased machine learning estimator for the cause-$j$-specific cumulative incidence function can be constructed as described in Section~\ref{sec: estimation and inference}. The estimator retains the same multiple robustness property and asymptotic linearity, and both pointwise and uniform inference can be conducted in exactly the same manner as in Section~\ref{sec: estimation and inference}.

%% file: 6.Simulation.tex
\section{Simulation}
\label{sec: simulation}
\subsection{Simulation setup}
We evaluated the finite-sample performance of the proposed estimator under various data generating processes. We generated the historical trial dataset $\mathcal{D}_h$ as follows:
\begin{description}
   \item[Baseline covariates $\boldsymbol{X}$ in $\mathcal{D}_h$] Baseline covariates in the historical trial are 6-dimensional and follow a multivariate distribution with $\mathbf{\mu} = (c, c, c, 0.8c, 0.8c, 0.8c)^T$, where $c$ is to be specified later, and $\Sigma = 0.5\cdot\textbf{I}_{6}$, where $\textbf{I}_{d}$ denotes a dimension-$d$ identity matrix.

   \item[Treatment assignment $A$ in $\mathcal{D}_h$] Assignment $A \in \{0, 1\}$ is randomized with probability $0.5$.

   \item[Peak immune biomarker $S$ in $\mathcal{D}_h$] The peak immune biomarker $S$ follows:
   \begin{equation*}
       S = 2 + 0.5 X_1 - X_2 + 1.5 X_3 + A \cdot (X_2 - 0.5 X_4 + 1.5 X_5 - X_6) + \epsilon,
   \end{equation*}
   where $\epsilon \sim \mathcal{N}(0, 1)$ follows a standard normal distribution.

   \item[Time-to-event outcome $T$ in $\mathcal{D}_h$] Participants' time-to-event outcome was generated from an exponential distribution with rate parameter:
   \begin{equation*}
       \lambda_T = 0.1\cdot\exp\left\{0.5X_1 - 0.5X_5 + 0.3S\cdot(1.3X_2 + 0.4X_4) + A\cdot(0.2X_2 + 0.6X_3 - 1.2X_6)\right\}.
   \end{equation*}
   \item[Time-to-censoring $C$ in $\mathcal{D}_h$] Participants' time-to-censoring was generated from an exponential distribution with rate parameter: $\lambda_C = 0.03\cdot \exp\left\{-0.4X_2 + 0.1S\right\}.$
\end{description}

Under this data-generating process, a participant’s peak immune marker level $S$ depends on treatment assignment $A$ and baseline covariates $\boldsymbol{X}$; the time-to-event $T$ depends on $\boldsymbol{X}$ and $S$; and the censoring mechanism is covariate-dependent.

In addition to $\mathcal{D}_h$, an immunobridging dataset $\mathcal{D}_b$ was further generated as follows:
\begin{description}
   \item[Baseline covariates in $\mathcal{D}_b$] Baseline covariates in the immunobridging dataset are also 6-dimensional and follow a multivariate distribution with $\mathbf{\mu} = (c, c, c, 0.8c, 0.8c, 0.8c)^T$, where $c$ is to be specified later, and $\Sigma = 0.5\cdot\textbf{I}_{6}$.

   \item[Treatment assignment in $\mathcal{D}_b$] The treatment assignment $A \in \{1, 1'\}$ is randomized with probability $0.5$ in the immunobridging study.

   \item[Peak immune biomarker $S$ in $\mathcal{D}_b$] The peak immune biomarker $S$ follows:
   \begin{equation*}
       S = 4 + 1.5X_1 - X_3 + 1.5X_6 + (\mathbbm{1}\{A = 1\} + 2\cdot\mathbbm{1}\{A = 1'\})\cdot(0.5X_2 + 0.5X_3 + 0.5X_4 - X_6) + \epsilon,
   \end{equation*}
   where $\epsilon$ follows a standard normal distribution.
\end{description}

According to this data-generating process, the peak immune marker in the immunobridging study also depends on both the vaccine and baseline covariates, and this relationship is different between the historical trial and the immunobridging study.

The parameter $c$ in the data-generating process of $\boldsymbol{X}$ controls the degree of overlap of the covariates $\boldsymbol{X}$ between $\mathcal{D}_h$ and $\mathcal{D}_b$. The first factor we vary is the parameter $c$:
\begin{description}
  \item[Factor 1: overlap parameter \(c\).] We consider four values of \(c\): \(c = 0\) (full overlap), \(c = 0.05\) (high overlap), \(c = 0.125\) (moderate overlap), and \(c = 0.25\) (poor overlap).
\end{description}

\noindent The second factor we vary is the sample size of historical trial dataset, $n_h$:
\begin{description}
    \item[Factor 2: sample size $n_h$ of $\mathcal{D}_h$.] We consider four choices of $n_h$: $1000$, $2000$, $3000$, and $4000$.  
\end{description}

For each $n_h$, we set the sample size of the immunobridging trial dataset to $n_b = n_h/4$. The target parameter is $\theta=P(T(A = 1') > 5)$, where $t = 5$ is chosen to correspond approximately to the maximum follow-up times. For each data generating process specified by the overlapping parameter $c$ and the sample size $n_h$, we calculated the efficient influence function-based DML estimator $\widehat{\theta}$ and the associated pointwise confidence interval. When estimating nuisance functions, the estimator adopted ensemble machine learning consisting of (1) generalized linear model with interaction, random forest, and generalized additive model for regression and classification; (2) Kaplan-Meier estimator, Cox regression with interaction, and generalized additive model for modeling the survival and censoring processes. These machine learning tools were implemented using the packages $\textsf{SuperLearner}$ \citep{Polley2024superlearner} and $\textsf{survSuperLearner}$ \citep{survSuperLearner} in $\textsf{R}$.


\subsection{Simulation results}
\label{subsec: simulation results}
Figure \ref{fig: simu res t = 5 ML} shows the sampling distributions of the proposed estimators for $P(T(A = 1') > 5)$ across $4 \times 4 = 16$ data-generating processes. The true parameter values are indicated in each panel as a dashed red line. Figure \ref{fig: simu res t = 5 ML} suggest the proposed estimator is approximately normally distributed and centered closely around the true estimand values. 

Table~\ref{tbl: simu res ML} further summarizes the mean, median, bias, percentage bias, root mean squared error, average standard error, and coverage of the estimated 95\% confidence intervals for the proposed estimator across different settings. Across the scenarios considered, the proposed estimator exhibits minimal bias, almost always below 1\%, even for moderately large sample sizes (e.g., \(n_h = 2000\) and \(n_b = 500\)). In most settings, the empirical coverage of the 95\% confidence interval is close to the nominal level.

\begin{figure}[ht]
    \centering
\includegraphics[width=0.75\linewidth]{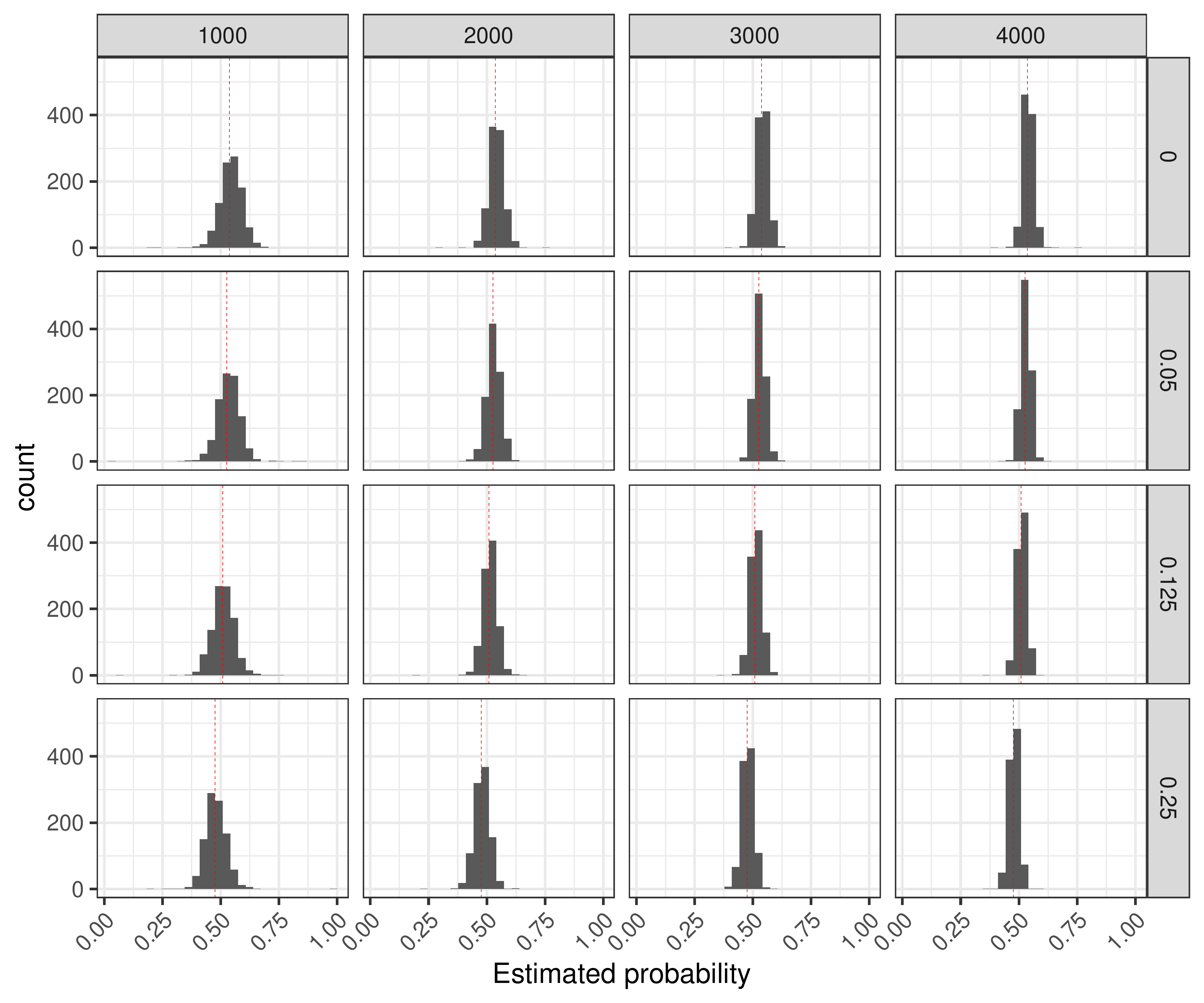}
    \caption{Sampling distributions of $\widehat{\theta}$ for $P(T(A = 1') > 5)$ corresponding to different data-generating processes specified by the sample size $n_h$ (columns) and the overlapping parameter $c$ (rows).  The red dashed lines correspond to the true parameter values.}
    \label{fig: simu res t = 5 ML}
\end{figure}

\begin{table}[ht]
\centering
\caption{Mean, median, bias, percentage of bias, root mean squared error, average standard error, and coverage of the proposed DML estimator for $P(T(A = 1') > 5)$ across a range of data-generating processes, characterized by the overlap parameter $c$ and the historical clinical trial sample size $n_h$, setting the the immunobridging trial sample size $n_b = n_h/4$. Simulations were repeated $1000$ times. }
\label{tbl: simu res ML}
\resizebox{0.8\textwidth}{!}{
\begin{tabular}{ccccccccccccccc}
  \hline
 $n_h$ & $c$ & \textbf{Mean} & \textbf{Median} &\textbf{Bias} &$\%$ \textbf{Bias} & \textbf{RMSE} & \textbf{Ave. SE} & \textbf{Coverage}\\ 
  \hline
1000 & 0.00 & 0.54 & 0.55 & 0.007 & 1.29\% & 0.049 & 0.049 & 96.6\% \\ 
  1000 & 0.05 & 0.53 & 0.53 & 0.008 & 1.48\% & 0.052 & 0.049 & 96.1\% \\ 
  1000 & 0.125 & 0.51 & 0.51 & 0.002 & 0.33\% & 0.049 & 0.047 & 95.8\% \\ 
  1000 & 0.25 & 0.48 & 0.48 & 0.004 & 0.80\% & 0.049 & 0.045 & 96.5\% \\ 
  2000 & 0.00 & 0.54 & 0.54 & 0.005 & 0.90\% & 0.033 & 0.033 & 96.2\% \\ 
  2000 & 0.05 & 0.53 & 0.53 & 0.002 & 0.46\% & 0.031 & 0.032 & 95.6\% \\ 
  2000 & 0.125 & 0.51 & 0.51 & 0.004 & 0.84\% & 0.033 & 0.032 & 96.8\% \\ 
  2000 & 0.25 & 0.48 & 0.48 & 0.003 & 0.70\% & 0.034 & 0.030 & 93.5\% \\ 
  3000 & 0.00 & 0.54 & 0.54 & 0.004 & 0.77\% & 0.026 & 0.026 & 96.3\% \\ 
  3000 & 0.05 & 0.53 & 0.53 & 0.003 & 0.57\% & 0.025 & 0.026 & 96.3\% \\ 
  3000 & 0.125 & 0.51 & 0.51 & 0.005 & 0.90\% & 0.026 & 0.025 & 93.8\% \\ 
  3000 & 0.25 & 0.48 & 0.48 & 0.003 & 0.56\% & 0.025 & 0.024 & 94.5\% \\ 
  4000 & 0.00 & 0.54 & 0.54 & 0.003 & 0.61\% & 0.024 & 0.023 & 94.9\% \\ 
  4000 & 0.05 & 0.53 & 0.53 & 0.003 & 0.66\% & 0.021 & 0.022 & 96.2\% \\ 
  4000 & 0.125 & 0.51 & 0.51 & 0.003 & 0.58\% & 0.022 & 0.022 & 95.8\% \\ 
  4000 & 0.25 & 0.48 & 0.48 & 0.003 & 0.60\% & 0.022 & 0.021 & 93.7\% \\ 
   \hline
\end{tabular}}
\end{table}

%% file: 7.Case-Study.tex
\section{The COVID-19 Variant Immunologic Landscape (COVAIL) trial}
\label{sec: case study COVAIL}
\subsection{Stages 2 and 4 of COVAIL}
\label{subsec: case study stages 2 and 4}
We illustrate Task I and our proposed method using Stages 2 and 4 of the COVAIL trial, conducted when Omicron BA.5 was the predominant circulating strain, with $S$ defined as the neutralizing antibody titer measured against the Omicron BA.4/5 strain. In Stage 2 of COVAIL, 313 participants were enrolled and randomized to receive one of six Pfizer-BioNTech mRNA vaccines, with the different vaccines defined by whether a single or two SARS-CoV-2 strains were represented in the vaccine construct (i.e., monovalent or bivalent), and the lineage of the strain: monovalent Prototype, monovalent Omicron, monovalent Beta, bivalent Beta + Omicron, bivalent Beta + Prototype, and bivalent Prototype + Omicron. Here, the Prototype lineage refers to the original ancestral lineage that circulated in early 2020, during the initial phase of the COVID-19 pandemic. We consider the $n = 154$ Stage-2 participants randomized to an Omicron-containing Pfizer-BioNTech mRNA vaccine (monovalent Omicron, bivalent Beta + Omicron, or bivalent Prototype + Omicron) and the $n = 47$ Stage-2 participants randomized to the monovalent Prototype Pfizer-BioNTech mRNA vaccine in the per-protocol correlates cohort. Because the sample size within each individual arm is limited, we follow \citet{zhang2025neutralizing} and pool the monovalent Omicron, bivalent Beta + Omicron, and bivalent Prototype + Omicron arms into a single Omicron-containing group.


\begin{figure}[!htb]
    \centering
    \begin{minipage}{.5\textwidth}
        \centering
    \includegraphics[width=0.97\linewidth]{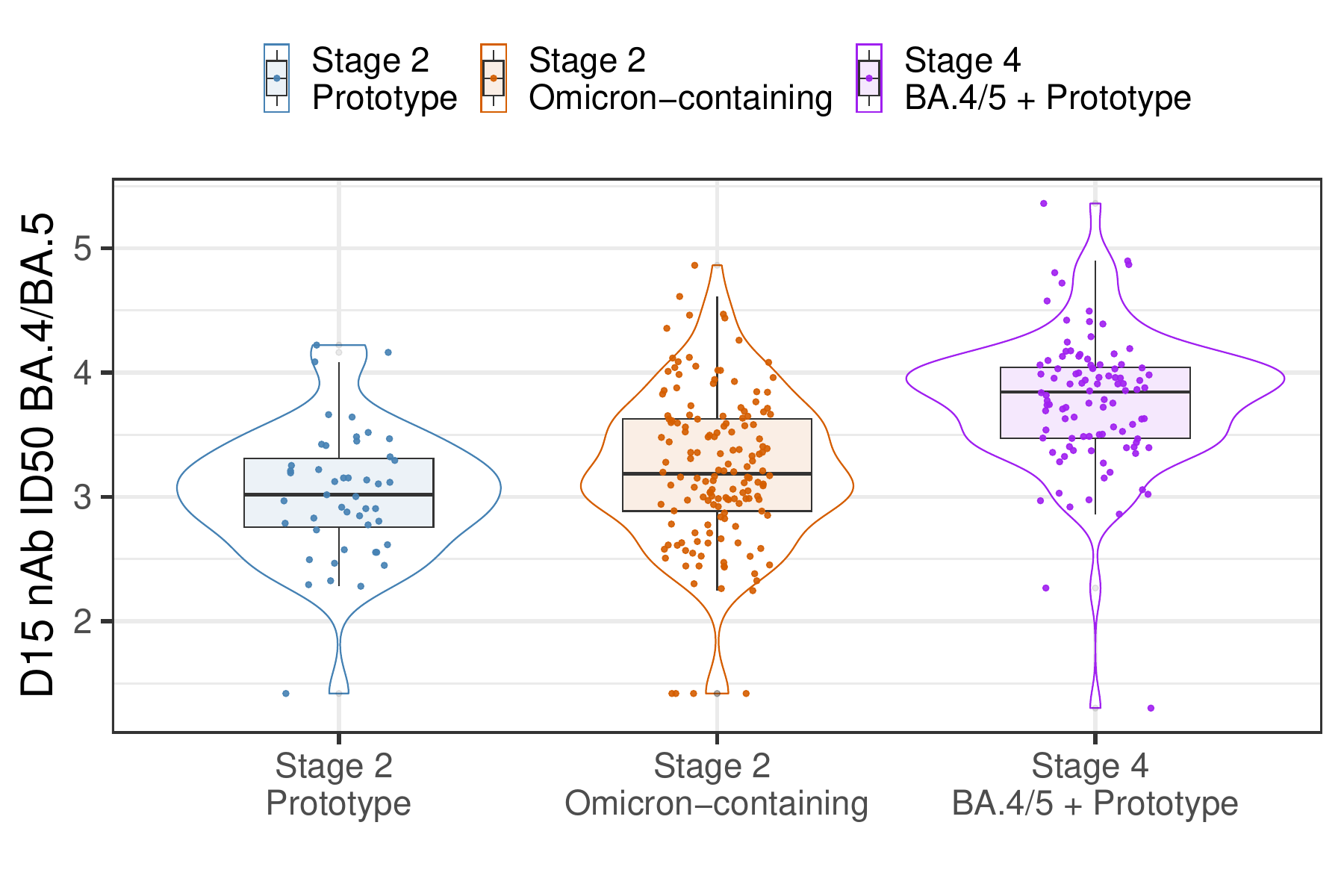}
      
    \end{minipage}%
    \begin{minipage}{0.5\textwidth}
        \centering
    \includegraphics[width=0.97\linewidth]{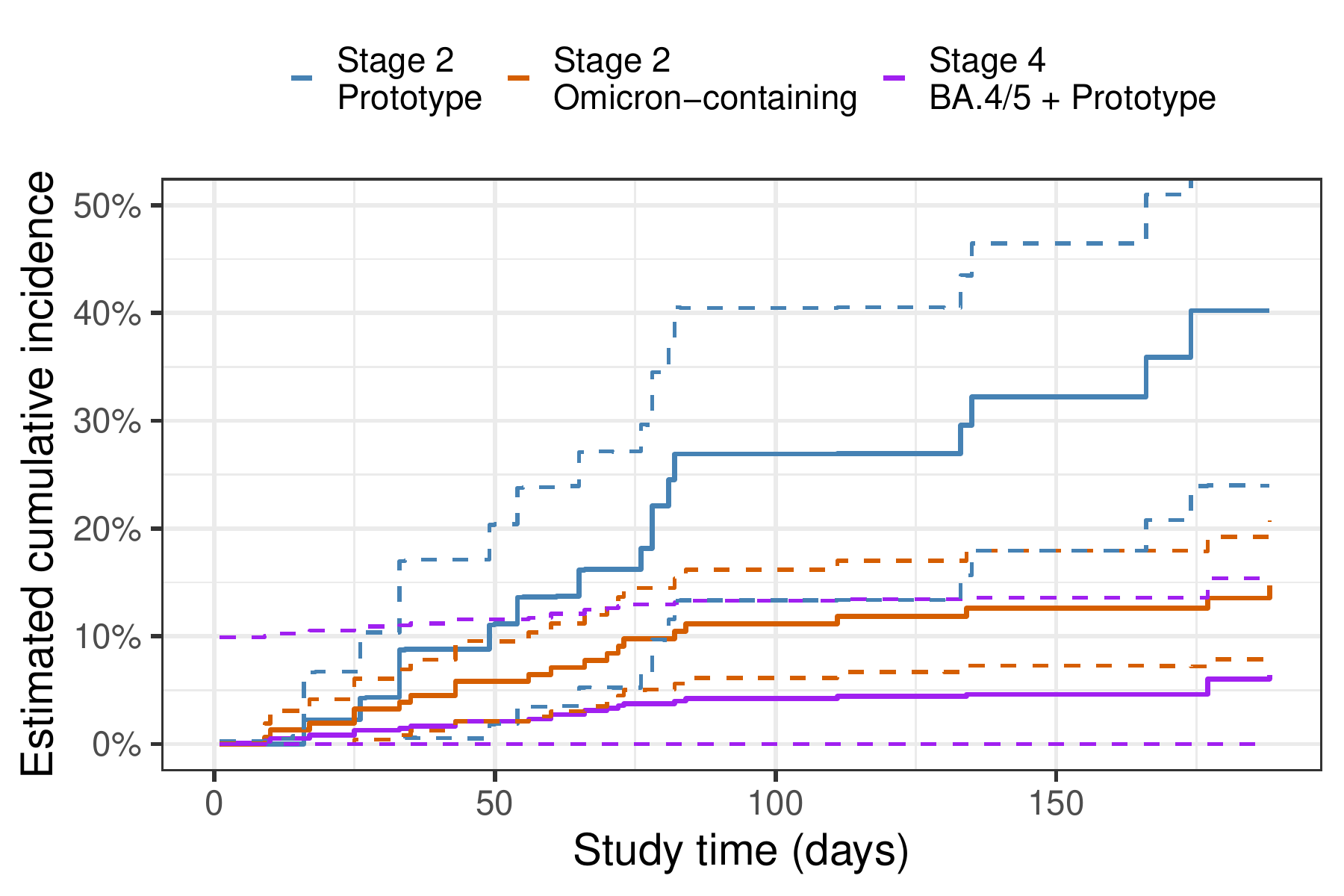}
      
    \end{minipage}
    \caption{\small\textbf{Left panel}: Boxplots and violin plots of Day 15 (D15) neutralizing antibody titers against Omicron BA.4/BA.5 among Stage-2 Prototype Pfizer-BioNTech vaccinee (blue), Stage-2 Omicron-containing Pfizer-BioNTech vaccinees (orange), and Stage-4 BA.4/5 + Prototype Pfizer-BioNTech vaccinees (purple). \textbf{Right panel}: Cumulative incidence curves for Stage-2 Prototype Pfizer-BioNTech vaccinees (blue) and Stage-2 Omicron-containing Pfizer-BioNTech vaccinees (orange), along with the estimated cumulative incidence curve for Stage-4 BA.4/5 + Prototype Pfizer-BioNTech vaccinees (purple). The cumulative incidence curve for Stage-4 BA.4/5 + Prototype Pfizer-BioNTech vaccinees (purple) was inferred using the proposed method based on the neutralizing antibody titer and COVID-19 outcome data of Stage-2 Omicron-containing Pfizer-BioNTech vaccinees ($\mathcal{D}_h$) as well as the neutralizing antibody titer data of Stage-4  BA.4/5 + Prototype Pfizer-BioNTech vaccinees ($\mathcal{D}_b$). Dashed lines indicate 95\% pointwise confidence intervals.}
    \label{fig: real data Pfizer}
\end{figure}

The first two boxplots and violin plots in the left panel of Figure \ref{fig: real data Pfizer} show the distribution of Day 15 neutralizing antibody titers against Omicron BA.4/BA.5 (nAb-ID50 BA.4/BA.5), measured approximately 15 days after receipt of the study vaccine, among Stage-2 Prototype Pfizer-BioNTech vaccinees (blue) and Stage-2 Omicron-containing Pfizer-BioNTech vaccinees (orange). Consistent with \citet{branche2023comparison}, Omicron-containing vaccines elicited higher nAb-ID50 titers against Omicron variants.

The clinical endpoint of the study was defined as a self-reported positive SARS-CoV-2 test (RT-PCR or antigen) or a study-conducted positive SARS-CoV-2 test, with the onset date taken as the earliest positive test date \citep{zhang2025neutralizing}. Among Stage-2 participants, 17 of 47 recipients of the Prototype vaccine experienced a clinical endpoint by 188 days post Day 15, compared to 21 of 154 recipients of the Omicron-containing vaccine within the same follow-up period. Figure S1 in Supplemental Material B presents the Kaplan–Meier curves along with the corresponding risk tables. The right panel of Figure \ref{fig: real data Pfizer} shows the estimated cumulative incidence curves for Stage-2 Prototype Pfizer–BioNTech vaccine recipients (blue) and Stage-2 Omicron-containing Pfizer–BioNTech vaccine recipients (orange), obtained using the method of \citet{westling2024inference}.


Later, in Stage 4 of the COVAIL trial, approximately $200$ participants were enrolled and randomized to receive either the bivalent BA.1 + Prototype Pfizer-BioNTech booster or the bivalent BA.4/5 + Prototype Pfizer-BioNTech booster. The purple boxplot and violin plot in the left panel of Figure \ref{fig: real data Pfizer} show the distribution of Day 15 nAb-ID50 BA.4/BA.5 titers among Stage-4 BA.4/5 + Prototype Pfizer-BioNTech vaccinees, measured using the same validated assay as that used for participants in Stage 2 of the trial. Given the full immunogenicity and clinical data for the Pfizer-BioNTech vaccines collected during Stage 2 of the trial, along with the immunogenicity data for the BA.4/5 + Prototype Pfizer-BioNTech vaccine, can we infer the counterfactual cumulative incidence curve that would have been observed if Stage-4 participants had received the BA.4/5 + Prototype Pfizer-BioNTech booster during Stage 2?

To address this question, we used neutralizing antibody titer and COVID-19 outcome data from Stage-2 Omicron-containing vaccinees ($\mathcal{D}_h$) together with Stage-4 BA.4/5 + Prototype immunogenicity data ($\mathcal{D}_b$). We assume that, conditional on baseline na\"ive status of being diagnosed as previously infected with SARS-CoV-2 and risk score, there is no controlled direct effects when comparing the Stage-2 Omicron-containing Pfizer-BioNTech vaccines with the Stage-4 BA.4/5 + Prototype Pfizer-BioNTech vaccine. The right panel of Figure \ref{fig: real data Pfizer} displays the estimated cumulative incidence curve obtained using our proposed EIF-based estimator. All nuisance functions were estimated using an ensemble machine-learning approach, as in the simulation studies. The target population for this curve is the Stage-4 study population, and the analysis adjusts for the baseline risk score and baseline naïve/non-naïve status.

Using our proposed estimator, the cumulative incidence is estimated to be 4.8\% (95\% CI: 0 to 13.8\%) at Day 91 (approximately 3 months post D15) and 6.8\% (95\% CI: 0 to 15.9\%) at Day 188 (approximately 6 months post D15). Figure S2 in Supplemental Material B complements these pointwise confidence intervals with a uniformly valid confidence band.

\subsection{Evaluating the key assumption of no controlled direct effects}
\label{subsec: case study assess immunobridging assumption}
Under our proposed estimation framework, a key assumption is the absence of controlled direct effects when comparing two vaccines (e.g., the approved vaccine versus the candidate vaccine). Here, we assess whether this assumption holds for the Stage-2 Prototype Pfizer-BioNTech vaccine and the Stage-2 Omicron-containing vaccine when the immune marker is the Day 15 neutralizing antibody titer against BA.4/BA.5, conditional on the baseline risk score and na\"ive/non-na\"ive status:
\begin{equation}\label{eq: case study IB assumption}
    P(T(A = \text{Prototype}, S = s) \mid \mathbf{X}) = P(T(A = \text{Omicron-containing}, S = s) \mid \mathbf{X}),
\end{equation}
where $S$ denotes the nAb-ID50 BA.4/BA.5 titer, and $\mathbf{X}$ consists of baseline risk score and na\"ive/non-na\"ive status. In this application, all immunogenicity data for Omicron-containing vaccinees, as well as the full immunogenicity and clinical data, were collected during Stage 2 of the trial; therefore, we do not differentiate the trial context using $\Gamma$.

To evaluate this assumption, we proceed in two steps. First, we estimate the counterfactual cumulative incidence curve when receiving the Omicron-containing vaccine, defined as
\[
F(t) := \mathbb{E}_{\mathcal{P}_{\text{stage-2}}}\left[P(T(A = \text{Omicron-containing}) \leq t \mid \mathbf{X})\right],
\]
where $\mathcal{P}_{\text{stage-2}}$ denotes the marginal distribution of $\mathbf{X}$ in Stage 2, using the neutralizing antibody titer and COVID-19 outcome data from Stage-2 Prototype vaccinees ($\mathcal{D}_h$) together with the neutralizing antibody titer data from Stage-2 Omicron-containing vaccinees ($\mathcal{D}_b$) under the assumption of no controlled direct effects. Then we compare $F(188)$, the counterfactual cumulative incidence at Day 188 (approximately 6 months post-vaccination), a prespecified time point in the COVAIL correlates study protocol \citep{zhang2025neutralizing}, with the actual cumulative incidence $F_{\text{actual}}(188)$ among Stage-2 Omicron-vaccinees. A significant difference between $F(188)$ and $F_{\text{actual}}(188)$ provides evidence against the assumption in Equation \eqref{eq: case study IB assumption}.

The purple curve in Figure \ref{fig: case study stage 2 actual and projected incidence} displays the counterfactual cumulative incidence under the assumption of no controlled direct effects comparing Omicron-containing and Prototype vaccines. At Day 188, the counterfactual cumulative incidence estimated under the assumption of no controlled direct effects is 31.8\% (95\% CI: 10.8\% to 51.4\%), whereas the actual cumulative incidence was lower at 14.5\% (95\% CI: 8.4\% to 20.5\%). Using bootstrap resampling, the 95\% confidence interval for the difference between the counterfactual cumulative incidence and the actual cumulative incidence at Day 188 was estimated to be [9.3\%, 27.5\%], providing evidence against the null hypothesis of no controlled direct effects when comparing the Prototype and Omicron-containing vaccines. This suggests that the Stage-2 Omicron-containing Pfizer-BioNTech vaccines may have influenced clinical outcomes through pathways not fully captured by the Day 15 neutralizing antibody titer. 

\begin{figure}[ht]
    \centering
\includegraphics[width=0.7\linewidth]{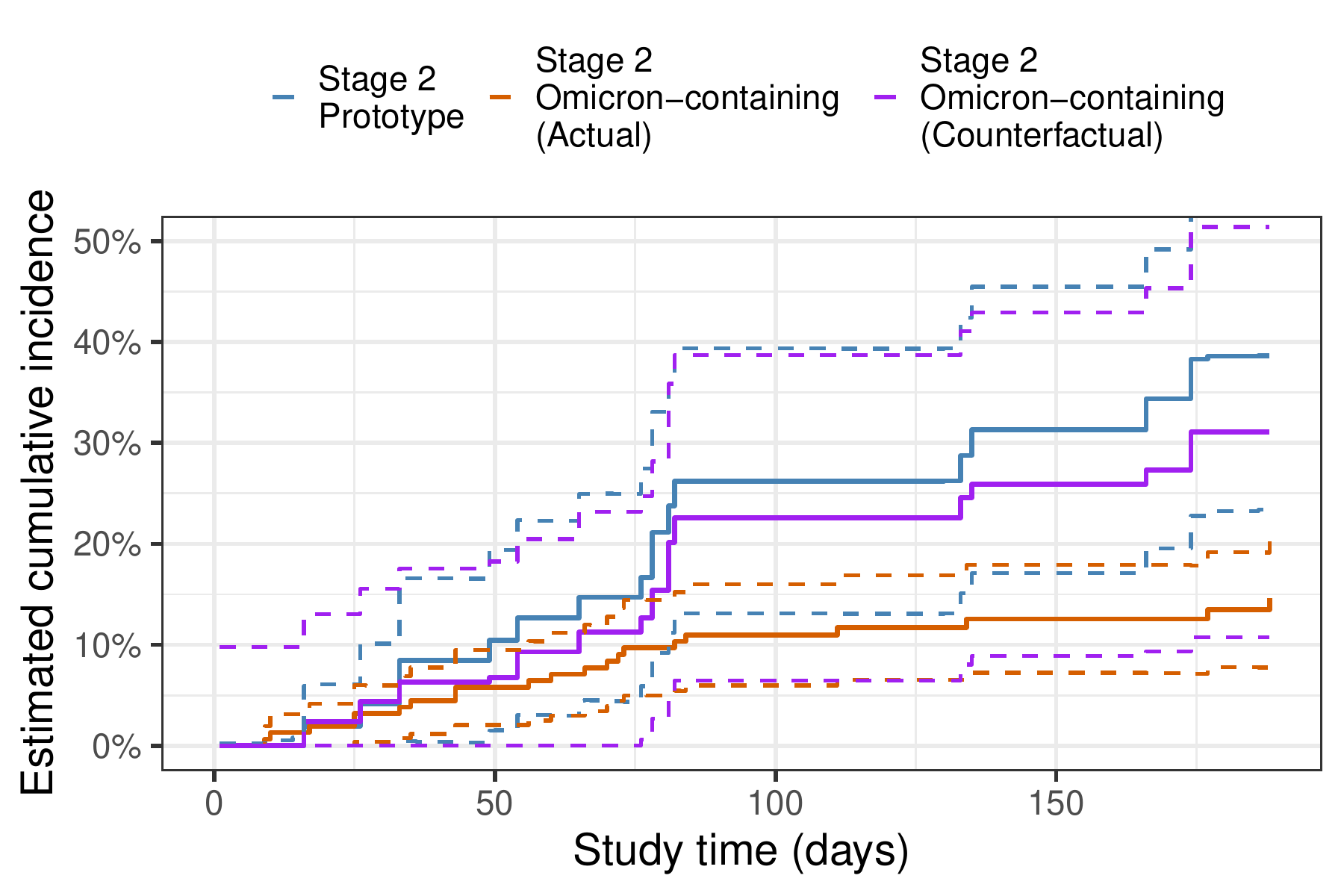}
    \caption{\small Cumulative incidence curves of Stage-2 Prototype Pfizer-BioNTech vaccinees (blue) and Stage-2 Omicron-containing Pfizer-BioNTech vaccinees (orange), along with the counterfactual cumulative incidence curve for Stage-2 Omicron-containing Pfizer-BioNTech vaccinees estimated under the no controlled direct effects assumption (purple). Dashed lines indicate 95\% pointwise confidence intervals.}
    \label{fig: case study stage 2 actual and projected incidence}
\end{figure}

%% file: 8.Discussion.tex
\section{Discussion}
\label{sec: discussion}
In this paper, we systematically study a central question in immunobridging studies: how to estimate counterfactual cumulative incidence curves and relative vaccine efficacy by integrating historical clinical trial data with new immunobridging studies. Our main contribution is to extend the methods of \citet{athey2025surrogate} and \citet{gilbert2024surrogateendpointbasedprovisional} to the estimation of cumulative incidence curves for survival endpoints and, more broadly, to the estimation of cause-specific cumulative incidence curves that accommodate the co-circulation of multiple viral strains or pathogen serotypes. 

We outline three common contexts in which immunobridging studies are conducted and, within each context, discuss in detail the identification assumptions that facilitate immunobridging. Because biological and immunologic systems are highly complex, identification assumptions should be regarded as a base-case framework for inference rather than definitive truth. Among identification assumptions, the overlap assumption technically need not be imposed: it can always be satisfied by defining the target trial population---characterized by the distribution of baseline covariates---as a subset of the historical trial population. If the target population is intrinsically different from the historical trial population---for example, when the historical trial is conducted in adults and the immunobridging study aims to bridge to children or adolescents---then this distinction must be made explicit and acknowledged as a limitation.

Conditional exchangeability between the two trial settings is the second key assumption. It requires that the potential time-to-event outcomes under the same joint intervention remain stable between the historical trial and the hypothetical trial. The “hypothetical trial setting” can be viewed as a continuation of the immunobridging study: rather than collecting only post-vaccination immunogenicity data, the trial would continue to accrue clinical endpoints. This exchangeability would be violated if the hypothetical trial were conducted during a period with a substantially different external force of infection than that of the historical trial. In such cases, the resulting estimator remains meaningful but has a different interpretation: instead of estimating the counterfactual cumulative incidence under the new vaccine in a hypothetical trial, it estimates what would have occurred had the new investigational vaccine, rather than the approved one, been used in the historical trial.

Arguably the most important assumption enabling immunobridging is the assumption of no controlled direct effects conditional on observed baseline covariates. To bring this assumption closer to reality, it may be helpful to incorporate multiple peak immune markers into a surrogate index as the utilized surrogate in the analysis, as advocated by \citet{athey2025surrogate} and \citet{gilbert2024surrogateendpointbasedprovisional}. Because it is unlikely that all mediators of a vaccine’s effect are fully captured and included in the analysis in making the no controlled direct effects assumption hold, this assumption is best viewed as a ``leading case.” By varying the magnitude of the controlled direct effects in a sensitivity analysis, the estimator obtained under the no controlled direct effects assumption can serve as an anchor point. 

While the present paper establishes identification under an assumption of no controlled direct effects between the approved and investigational vaccines, alternative assumptions that preserve other features of the causal mechanism could be considered and may be more appropriate in certain settings. Exploring such alternatives represents an important direction for future work. 

A limitation of the present work is its exclusive focus on immune markers measured at the prespecified peak time point. Consequently, the method cannot distinguish between investigational vaccines that elicit similar peak immune responses. For example, if one vaccine induces a more durable response than another, we would expect differences in cumulative incidence despite similar peak responses. Future work will extend the framework to immunobridging settings in which immune markers are measured longitudinally in both the historical trial and the immunobridging trial. 

%% file: SM.Proofs.tex
\section{Proofs}

\subsection{Proofs of Proposition~\ref{prop:identification-complete-data}, \ref{prop: identification Task II} and Theorem~\ref{thm:identification-complete-data competing risk}}
\label{proof-prop:identification-complete-data}

We prove the first identification result via the ``mediation causal formula" \citep{pearl2012mediation}.

\begin{proof}
For $a \in \{1, 1'\}$, we have
\begin{align*}
R(a; \Gamma = 1) &= \mathbb{E}_{\mathcal{P}_b}\left[\mathbb{E}\{y(T(a)) \mid \boldsymbol{X}, \Gamma = 1\}\right] \\
&= \mathbb{E}_{\mathcal{P}_b}\left[\mathbb{E}\{y(T) \mid \boldsymbol{X}, A = a, \Gamma = 1\}\right] \\
&= \mathbb{E}_{\mathcal{P}_b}\left[\int_{s\in \mathcal{S}} \mathbb{E}\{y(T) \mid \boldsymbol{X}, A = a, S = s, \Gamma = 1\} f(s \mid \boldsymbol{X}, A = a, \Gamma = 1) ds\right] \\
&= \mathbb{E}_{\mathcal{P}_b}\left[\int_{s\in \mathcal{S}} \mathbb{E}\{y(T) \mid \boldsymbol{X}, A = 1, S = s, \Gamma = 0\} f(s \mid \boldsymbol{X}, A = a, \Gamma = 1) ds\right] \\
&= \mathbb{E}\left[\frac{\Gamma}{\kappa} \int_{s\in \mathcal{S}} \mu( \boldsymbol{X}, 1, s) f(s \mid  \boldsymbol{X}, A = a, \Gamma = 1) ds\right].
\end{align*}
\end{proof}

We can then prove the second identification result via the following alternative representation.
\begin{proof}
For $a \in \{1, 1'\}$, we have
\begin{align*}
&\mathbb{E}\left[\frac{\Gamma}{\kappa} \int_{s\in \mathcal{S}} \mu( \boldsymbol{X}, 1, s) f(s \mid  \boldsymbol{X}, A = a, \Gamma = 1) ds\right] \\
& = \iint_{\mathcal{X} \times \mathcal{S}} \mu( \boldsymbol{x}, 1, s) dF(s \mid \boldsymbol{X} = \boldsymbol{x}, A = a, \Gamma = 1) dF(\boldsymbol{x} \mid \Gamma = 1) \\
& = \sum_{a' \in \{1, 1'\}} \iint_{\mathcal{X} \times \mathcal{S}} \mu( \boldsymbol{x}, 1, s) \frac{I\{A = a\}}{f(A = a' \mid \boldsymbol{X} = \boldsymbol{x}, \Gamma = 1)} dF(s, \boldsymbol{x}, a' \mid \Gamma = 1) \\
& = \mathbb{E}\left[\frac{\Gamma I\{A = a\}}{\kappa f(A = a \mid \boldsymbol{X}, \Gamma = 1)} \mu(\boldsymbol{X}, 1, S)\right].
\end{align*}
\end{proof}

Finally, we prove the third identification result via a weighting strategy.
\begin{proof}
For $a \in \{1, 1'\}$, we have
{\scriptsize
\begin{align*}
&\mathbb{E}\left[\frac{\Gamma}{\kappa} \int_{s\in \mathcal{S}} \mu( \boldsymbol{X}, 1, s) f(s \mid  \boldsymbol{X}, A = a, \Gamma = 1) ds\right] \\
& = \sum_{a' \in \{1, 1'\}} \iiint_{\mathcal{X} \times \mathcal{S} \times \mathcal{T}} \frac{(1 - \Gamma) f(\Gamma = 1 \mid \boldsymbol{X} = \boldsymbol{x})}{\kappa} \frac{I\{A = 1\}}{f(A = 1, \Gamma = 0 \mid \boldsymbol{X} = \boldsymbol{x})} \frac{f(s \mid \boldsymbol{X} = \boldsymbol{x}, A = a, \Gamma = 1)}{f(s \mid \boldsymbol{X} = \boldsymbol{x}, A = 1, \Gamma = 0)} y(t) dF(y, s, \boldsymbol{x}, a') \\
& = \mathbb{E}\left[\frac{(1 - \Gamma) I\{A = 1\}}{\kappa} \frac{f(\Gamma = 1 \mid \boldsymbol{X})}{f(A = 1, \Gamma = 0 \mid \boldsymbol{X})} \frac{f(S \mid \boldsymbol{X}, A = a, \Gamma = 1)}{f(S \mid \boldsymbol{X}, A = 1, \Gamma = 0)} y(T)\right].
\end{align*}}
\end{proof}

Next, we consider the identification task II, where we apply the same strategies as above. The first identification result follows the mediation causal formula.
\begin{proof}
We have
\begin{align*}
R'(1'; \Gamma = 1) &= \mathbb{E}_{\mathcal{P}_b}\left[\mathbb{E}\{y(T'(1')) \mid \boldsymbol{X}, \Gamma = 1\}\right] \\
&= \mathbb{E}_{\mathcal{P}_b}\left[\mathbb{E}\{y(T') \mid \boldsymbol{X}, A = 1', \Gamma = 1\}\right] \\
&= \mathbb{E}_{\mathcal{P}_b}\left[\int_{s'\in \mathcal{S}} \mathbb{E}\{y(T') \mid \boldsymbol{X}, A = 1', S' = s', \Gamma = 1\} f(s' \mid \boldsymbol{X}, A = 1', \Gamma = 1) ds'\right] \\
&= \mathbb{E}_{\mathcal{P}_b}\left[\int_{s\in \mathcal{S}} \mathbb{E}\{y(T) \mid \boldsymbol{X}, A = 1, S = s, \Gamma = 0\} f(s \mid \boldsymbol{X}, A = 1', \Gamma = 1) ds\right] \\
&= \mathbb{E}\left[\frac{\Gamma}{\kappa} \int_{s\in \mathcal{S}} \mu( \boldsymbol{X}, 1, s) f(s \mid  \boldsymbol{X}, A = 1', \Gamma = 1) ds\right].
\end{align*}
\end{proof}
The remaining two identification results are
\begin{equation*}
    R'(1'; \Gamma = 1) = \mathbb{E}\left[\frac{\Gamma I\{A = 1'\}}{\kappa f(A = 1' \mid \boldsymbol{X}, \Gamma = 1)} \mu(\boldsymbol{X}, 1, S)\right],
\end{equation*}
and
\begin{equation*}
    R'(1'; \Gamma = 1) = \mathbb{E}\left[\frac{(1 - \Gamma) I\{A = 1\}}{\kappa} \frac{f(\Gamma = 1 \mid \boldsymbol{X})}{f(A = 1, \Gamma = 0 \mid \boldsymbol{X})} \frac{f(S \mid \boldsymbol{X}, A = 1', \Gamma = 1)}{f(S \mid \boldsymbol{X}, A = 1, \Gamma = 0)} y(T)\right],
\end{equation*}
the derivations of which follow an identical procedure and are thus omitted.

For complete competing risks data, the three identification results can be derived using the procedure described above by setting $y(T) = I\{T \leq t, \Delta = j\}, j = 1, \ldots, J$.

\subsection{Proof of Theorem~\ref{thm:identification-censored-data} and \ref{thm:identification-censored-data competing risks}}

The proofs of identification results for right-censored data essentially follow the same procedure given in Section~\ref{proof-prop:identification-complete-data}, with an additional step of inverse probability of censoring weighting (IPCW).
\begin{proof}
For right-censored outcomes, the IPCW method gives identification of the (transformed) survival time:
\begin{align*}
& \mathbb{E}\left[\frac{\Delta y(\Tilde{T})}{G^{C}(\Tilde{T} \mid \boldsymbol{X}, A = a, S)} \mid \boldsymbol{X}, A = a, S\right] \\
& = \mathbb{E}\left[\frac{I\{C > T(a)\}}{P\{C > T(a) \mid \boldsymbol{X}, A = a, S\}} y(T(a)) \mid \boldsymbol{X}, A = a, S\right] \\
& = \mathbb{E}\left[\mathbb{E}\left[\frac{I\{C > T(a)\}}{P\{C > T(a) \mid \boldsymbol{X}, A = a, S\}} \mid \boldsymbol{X}, A = a, S, T(a)\right] y(T(a)) \mid \boldsymbol{X}, A = a, S\right] \\
& = \mathbb{E}\left[y(T(a)) \mid \boldsymbol{X}, S\right],
\end{align*}
and the remaining proofs are identical to those in Section~\ref{proof-prop:identification-complete-data} and are thus omitted.
\end{proof}

The extension to censored competing risks data is straightforward by setting $y(T) = I\{T \leq t, \Delta = j\}, j = 1, \ldots, J$.

\subsection{Proof of Proposition~\ref{prop:EIF-complete-data} and Theorem~\ref{thm:EIF-complete-data competing risk}}

We use the \emph{point mass contamination} strategy advocated in \cite{hines2022demystifying} to derive efficient influence functions.

\begin{proof}
We consider the target parameter defined as
\begin{equation*}
\Psi(\mathcal{P}) = \mathbb{E}\left[\frac{\Gamma}{\kappa} \int_{s\in \mathcal{S}} \mu(\boldsymbol{X}, 1, s) f(s \mid \boldsymbol{X}, A = a, \Gamma = 1) ds\right].
\end{equation*}

Perturbing $\mathcal{P}$ in the direction of a point mass, we find
\begin{align*}
\Psi(\mathcal{P}_\epsilon) &= \iiint y(t) f_\epsilon(t \mid x, 1, s, 0) dt f_\epsilon(s \mid x, a, 1) ds f_\epsilon(x \mid 1) dx \\
&= \iiint y(t) \frac{f_\epsilon(t,x,1,s,0) f_\epsilon(s,x,a,1) f_\epsilon(x,1)}{f_\epsilon(x,1,s,0) f_\epsilon(x,a,1) f_\epsilon(1)} dtdsdx,
\end{align*}
where the parametric submodel is indexed by $\epsilon$.

By the chain rule, we thus have
\begin{align*}
\left.\frac{d\Psi(\mathcal{P}_\epsilon)}{d\epsilon} \right|_{\epsilon = 0} &= \iiint y(t) \left\{\frac{f(s,x,a,1) f(x,1)}{f(x,1,s,0) f(x,a,1) f(1)} \left.\frac{d}{d\epsilon} f_\epsilon(t,x,1,s,0) \right|_{\epsilon = 0} \right.\\
& \left.\quad + \frac{f(t,x,1,s,0) f(x,1)}{f(x,1,s,0) f(x,a,1) f(1)} \left.\frac{d}{d\epsilon} f_\epsilon(s,x,a,1) \right|_{\epsilon = 0} \right.\\
& \left.\quad + \frac{f(t,x,1,s,0) f(s,x,a,1)}{f(x,1,s,0) f(x,a,1) f(1)} \left.\frac{d}{d\epsilon} f_\epsilon(x,1) \right|_{\epsilon = 0} \right.\\
& \left.\quad - \frac{f(t,x,1,s,0) f(s,x,a,1) f(x,1)}{f^2(x,1,s,0) f(x,a,1) f(1)} \left.\frac{d}{d\epsilon} f_\epsilon(x,1,s,0) \right|_{\epsilon = 0} \right.\\
& \left.\quad - \frac{f(t,x,1,s,0) f(s,x,a,1) f(x,1)}{f(x,1,s,0) f^2(x,a,1) f(1)} \left.\frac{d}{d\epsilon} f_\epsilon(x,a,1) \right|_{\epsilon = 0} \right.\\
& \left.\quad - \frac{f(t,x,1,s,0) f(s,x,a,1) f(x,1)}{f(x,1,s,0) f(x,a,1) f^2(1)} \left.\frac{d}{d\epsilon} f_\epsilon(1) \right|_{\epsilon = 0} \right\} dtdsdx.
\end{align*}

Evaluating the above integrals, we find
\begin{align*}
& \iiint y(t) \frac{f(s,x,a,1) f(x,1)}{f(x,1,s,0) f(x,a,1) f(1)} \left.\frac{d}{d\epsilon} f_\epsilon(t,x,1,s,0) \right|_{\epsilon = 0} dtdsdx \\
&= \iiint y(t) \frac{f(s,x,a,1) f(x,1)}{f(x,1,s,0) f(x,a,1) f(1)} I\{\Tilde{t},\Tilde{x},1,\Tilde{s},0\} dtdsdx \\
&= y(\Tilde{t}) \frac{f(\Tilde{s},\Tilde{x},a,1) f(\Tilde{x},1)}{f(\Tilde{x},1,\Tilde{s},0) f(\Tilde{x},a,1) f(1)} I\{\Tilde{a} = 1,\Tilde{\gamma} = 0\} \\
&= I\{\Tilde{a} = 1,\Tilde{\gamma} = 0\} y(t) \frac{f(\Tilde{s} \mid \Tilde{x},a,1) f(1 \mid \Tilde{x})}{f(\Tilde{s} \mid \Tilde{x},1,0) f(1,0 \mid \Tilde{x}) f(1)},
\end{align*}
\begin{align*}
& \iiint y(t) \frac{f(t,x,1,s,0) f(x,1)}{f(x,1,s,0) f(x,a,1) f(1)} \left.\frac{d}{d\epsilon} f_\epsilon(s,x,a,1) \right|_{\epsilon = 0} dtdsdx \\
&= \iiint y(t) \frac{f(t,x,1,s,0) f(x,1)}{f(x,1,s,0) f(x,a,1) f(1)} I\{\Tilde{s},\Tilde{x},a,1\} dtdsdx \\
&= I\{\Tilde{a} = a,\Tilde{\gamma} = 1\} \frac{f(\Tilde{x}, 1)}{f(\Tilde{x}, a,1) f(1)} \int y(t) f(t \mid \Tilde{x}, 1, \Tilde{s}, 0) dt \\
&= \frac{I\{\Tilde{a} = a,\Tilde{\gamma} = 1\}}{f(a \mid \Tilde{x}, 1) f(1)} \mu(\Tilde{x},1,\Tilde{s}),
\end{align*}
\begin{align*}
& \iiint y(t) \frac{f(t,x,1,s,0) f(s,x,a,1)}{f(x,1,s,0) f(x,a,1) f(1)} \left.\frac{d}{d\epsilon} f_\epsilon(x,1) \right|_{\epsilon = 0} dtdsdx \\
&= \iiint y(t) \frac{f(t,x,1,s,0) f(s,x,a,1)}{f(x,1,s,0) f(x,a,1) f(1)} I\{\Tilde{x}, 1\} dtdsdx \\
&= \frac{I\{\Tilde{\gamma} = 1\}}{f(1)} \iint y(t) f(t \mid \Tilde{x}, 1,s,0) dt f(s \mid \Tilde{x}, a,1) ds \\
&= \frac{I\{\Tilde{\gamma} = 1\}}{f(1)} \int \mu(\Tilde{x},1,s) f(s \mid \Tilde{x}, a,1) ds,
\end{align*}
\begin{align*}
& \iiint y(t) \frac{f(t,x,1,s,0) f(s,x,a,1) f(x,1)}{f^2(x,1,s,0) f(x,a,1) f(1)} \left.\frac{d}{d\epsilon} f_\epsilon(x,1,s,0) \right|_{\epsilon = 0} dtdsdx \\
&= \iiint y(t) \frac{f(t,x,1,s,0) f(s,x,a,1) f(x,1)}{f^2(x,1,s,0) f(x,a,1) f(1)} I\{\Tilde{x},1,\Tilde{s},0\} dtdsdx \\
&= I\{\Tilde{a} = 1, \Tilde{\gamma} = 0\} \frac{f(\Tilde{s}, \Tilde{x}, a, 1) f(\Tilde{x}, 1)}{f(\Tilde{x}, 1, \Tilde{s}, 0) f(\Tilde{x}, a, 1) f(1)} \int y(t) f(t \mid \Tilde{x}, 1, \Tilde{s}, 0) dt \\
&= \frac{I\{\Tilde{a} = 1, \Tilde{\gamma} = 0\} f(\Tilde{s} \mid \Tilde{x}, a, 1) f(1 \mid \Tilde{x})}{f(1) f(\Tilde{s} \mid \Tilde{x}, 1, 0) f(1, 0 \mid \Tilde{x})} \mu(\Tilde{x}, 1, \Tilde{s}),
\end{align*}
\begin{align*}
& \iiint y(t) \frac{f(t,x,1,s,0) f(s,x,a,1) f(x,1)}{f(x,1,s,0) f^2(x,a,1) f(1)} \left.\frac{d}{d\epsilon} f_\epsilon(x,a,1) \right|_{\epsilon = 0} dtdsdx \\
&= \iiint y(t) \frac{f(t,x,1,s,0) f(s,x,a,1) f(x,1)}{f(x,1,s,0) f^2(x,a,1) f(1)} I\{\Tilde{x},a,1\} dtdsdx \\
&= I\{\Tilde{a}=a,\Tilde{\gamma}=1\} \frac{f(\Tilde{x}, 1)}{f(\Tilde{x}, a, 1) f(1)} \iint y(t) f(t \mid \Tilde{x},1,s,0) dt f(s \mid \Tilde{x},a,1) ds \\
&= \frac{I\{\Tilde{a}=a,\Tilde{\gamma}=1\}}{f(a \mid \Tilde{x}, 1) f(1)} \int \mu(\Tilde{x}, 1, s) f(s \mid \Tilde{x},a,1) ds,
\end{align*}
\begin{equation*}
\iiint y(t) \frac{f(t,x,1,s,0) f(s,x,a,1) f(x,1)}{f(x,1,s,0) f(x,a,1) f^2(1)} \left.\frac{d}{d\epsilon} f_\epsilon(1) \right|_{\epsilon = 0} dtdsdx = \Psi(\mathcal{P}),
\end{equation*}
which gives the canonical gradient of $\Psi(\mathcal{P})$:
\begin{align*}
\left.\frac{d\Psi(\mathcal{P}_\epsilon)}{d\epsilon} \right|_{\epsilon = 0} &= \frac{(1 - \Tilde{\gamma}) I\{\Tilde{a} = 1\}}{\kappa} \frac{f(\Tilde{\gamma} = 1 \mid \Tilde{x})}{f(\Tilde{a} = 1, \Tilde{\gamma} = 0 \mid \Tilde{x})} \frac{f(\Tilde{s} \mid \Tilde{x}, \Tilde{a} = a, \Tilde{\gamma} = 1)}{f(\Tilde{s} \mid \Tilde{x}, \Tilde{a} = 1, \Tilde{\gamma} = 0)} \left\{y(\Tilde{t}) - \mu(\Tilde{t}, 1, \Tilde{s})\right\} \\
& + \frac{\Tilde{\gamma} I\{\Tilde{a} = a\}}{\kappa f(\Tilde{a} = a \mid \Tilde{x}, \Tilde{\gamma} = 1)} \left\{\mu(\Tilde{x}, 1, \Tilde{s}) - \int_{\Tilde{s}\in \mathcal{S}} \mu(\Tilde{x}, 1, \Tilde{s}) f(\Tilde{s} \mid \Tilde{x}, \Tilde{a} = a, \Tilde{\gamma} = 1) d\Tilde{s}\right\} \\
& + \frac{\Tilde{\gamma}}{\kappa} \int_{\Tilde{s}\in \mathcal{S}} \mu(\Tilde{x}, 1, \Tilde{s}) f(\Tilde{s} \mid \Tilde{x}, \Tilde{a} = a, \Tilde{\gamma} = 1) d\Tilde{s} - \Psi(\mathcal{P}).
\end{align*}

Hence, we conclude that $\Psi(\mathcal{P})$ is pathwise differentiable with the above EIF.
\end{proof}
The extension to complete competing risks data is straightforward by setting $y(T) = I\{T \leq t, \Delta = j\}, j = 1, \ldots, J$.

\subsection{Proof of Proposition~\ref{prop:MR-complete-data}}
\label{proof-MR-complete-data}

\begin{proof}
Denote by $f^\ast$ and $\mu^\ast$ the probability limits of the nuisance function estimators $\hat{f}$ and $\hat{\mu}$. Essentially we need to show
\begin{equation}\label{eq:MR-complete-data-proof}
\begin{split}
& \mathbb{E}\left[\frac{(1 - \Gamma) I\{A = 1\}}{\kappa} \frac{f^\ast(\Gamma = 1 \mid \boldsymbol{X})}{f^\ast(A = 1, \Gamma = 0 \mid \boldsymbol{X})} \frac{f^\ast(S \mid \boldsymbol{X}, A = a, \Gamma = 1)}{f^\ast(S \mid \boldsymbol{X}, A = 1, \Gamma = 0)} \left\{y(T) - \mu^\ast(\boldsymbol{X}, 1, S)\right\}\right. \\
& \left.+ \frac{\Gamma I\{A = a\}}{\kappa f^\ast(A = a \mid \boldsymbol{X}, \Gamma = 1)} \left\{\mu^\ast(\boldsymbol{X}, 1, S) - \int_{s\in \mathcal{S}} \mu^\ast(\boldsymbol{X}, 1, s) f^\ast(s \mid \boldsymbol{X}, A = a, \Gamma = 1) ds\right\}\right. \\
& \left.+ \frac{\Gamma}{\kappa} \int_{s\in \mathcal{S}} \mu^\ast(\boldsymbol{X}, 1, s) f^\ast(s \mid  \boldsymbol{X}, A = a, \Gamma = 1) ds - \Psi(\mathcal{P})\right] = 0
\end{split}
\end{equation}
under the union model $\mathcal{M}_a \cup \mathcal{M}_b \cup \mathcal{M}_c$.

First note that $\mu^\ast(x, 1, s) = \mu(x, 1, s)$ and $f^\ast(s \mid x, A = a, \Gamma = 1) = f(s \mid x, A = a, \Gamma = 1)$ under model $\mathcal{M}_a$. Equation~\eqref{eq:MR-complete-data-proof} follows because
\begin{equation*}
\mathbb{E}\left[\frac{(1 - \Gamma) I\{A = 1\}}{\kappa} \frac{f^\ast(\Gamma = 1 \mid X)}{f^\ast(A = 1, \Gamma = 0 \mid X)} \frac{f(S \mid X, A = a, \Gamma = 1)}{f^\ast(S \mid X, A = 1, \Gamma = 0)} \left\{y(T) - \mu(X, 1, S)\right\}\right] = 0,
\end{equation*}
\begin{equation*}
\mathbb{E}\left[\frac{\Gamma I\{A = a\}}{\kappa f^\ast(A = a \mid X, \Gamma = 1)} \left\{\mu(X, 1, S) - \int_{s\in \mathcal{S}} \mu(\boldsymbol{X}, 1, s) f(s \mid \boldsymbol{X}, A = a, \Gamma = 1) ds\right\}\right] = 0,
\end{equation*}
and
\begin{equation*}
\mathbb{E}\left[\frac{\Gamma}{\kappa} \int_{s\in \mathcal{S}} \mu(\boldsymbol{X}, 1, s) f(s \mid \boldsymbol{X}, A = a, \Gamma = 1) ds\right] = \Psi(\mathcal{P}).
\end{equation*}

Second, we have $\mu^\ast(x, 1, s) = \mu(x, 1, s)$ and $f^\ast(A = a \mid x, \Gamma = 1) = f(A = a \mid x, \Gamma = 1)$ under model $\mathcal{M}_b$. Equation~\eqref{eq:MR-complete-data-proof} now follows because
\begin{equation*}
\mathbb{E}\left[\frac{(1 - \Gamma) I\{A = 1\}}{\kappa} \frac{f^\ast(\Gamma = 1 \mid X)}{f^\ast(A = 1, \Gamma = 0 \mid X)} \frac{f^\ast(S \mid X, A = a, \Gamma = 1)}{f^\ast(S \mid X, A = 1, \Gamma = 0)} \left\{y(T) - \mu(X, 1, S)\right\}\right] = 0,
\end{equation*}
\begin{equation*}
\mathbb{E}\left[\left\{\frac{\Gamma}{\kappa} - \frac{\Gamma I\{A = a\}}{\kappa f(A = a \mid X, \Gamma = 1)}\right\} \int_{s\in \mathcal{S}} \mu(\boldsymbol{X}, 1, s) f^\ast(s \mid \boldsymbol{X}, A = a, \Gamma = 1) ds\right] = 0,
\end{equation*}
and
\begin{equation*}
\mathbb{E}\left[\frac{\Gamma I\{A = a\}}{\kappa f(A = a \mid X, \Gamma = 1)} \mu(X, 1, S)\right] = \Psi(\mathcal{P}).
\end{equation*}

Third, as $f^\ast(A = 1' \mid x, \Gamma = 1) = f(A = 1' \mid x, \Gamma = 1)$, $f^\ast(\Gamma = 1 \mid x) = f(\Gamma = 1 \mid x)$, $f^\ast(A = 1, \Gamma = 0 \mid x) = f(A = 1, \Gamma = 0 \mid x)$, $f^\ast(s \mid x, A = 1', \Gamma = 1) = f(s \mid x, A = 1', \Gamma = 1)$ and $f^\ast(s \mid x, A = 1, \Gamma = 0) = f(s \mid x, A = 1, \Gamma = 0)$ under model $\mathcal{M}_c$, we conclude that Equation~\eqref{eq:MR-complete-data-proof} holds because
\begin{align*}
& \mathbb{E}\left[\left\{\frac{(1 - \Gamma) I\{A = 1\}}{\kappa} \frac{f(\Gamma = 1 \mid X)}{f(A = 1, \Gamma = 0 \mid X)} \frac{f(S \mid X, A = a, \Gamma = 1)}{f(S \mid X, A = 1, \Gamma = 0)} \right.\right.\\
& \quad\quad \left.\left.- \frac{\Gamma I\{A = a\}}{\kappa f(A = a \mid X, \Gamma = 1)}\right\} \mu^\ast(X, 1, S)\right] = 0,
\end{align*}
\begin{equation*}
\mathbb{E}\left[\left\{\frac{\Gamma}{\kappa} - \frac{\Gamma I\{A = a\}}{\kappa f(A = a \mid X, \Gamma = 1)}\right\} \int_{s\in \mathcal{S}} \mu^\ast(\boldsymbol{X}, 1, s) f^\ast(s \mid \boldsymbol{X}, A = a, \Gamma = 1) ds\right] = 0,
\end{equation*}
and
\begin{equation*}
\mathbb{E}\left[\frac{(1 - \Gamma) I\{A = 1\}}{\kappa} \frac{f(\Gamma = 1 \mid X)}{f(A = 1, \Gamma = 0 \mid X)} \frac{f(S \mid X, A = a, \Gamma = 1)}{f(S \mid X, A = 1, \Gamma = 0)} y(T)\right] = \Psi(\mathcal{P}),
\end{equation*}
which completes the proof of multiple robustness.
\end{proof}

\subsection{Proof of Proposition~\ref{prop: asymptotic linearity complete data}}

We first state a useful lemma from \citet{Kennedy2020sharp}.
\begin{lemma}\label{lemma:crossfitting}
Let $\hat{f}(o)$ be a function estimated from a sample $O^N = \left(O_{n+1}, \ldots, O_{N}\right)$, and let $\mathbb{P}_n$ denote the empirical measure over $\left(O_1, \ldots, O_n\right)$, which is independent of $O^N$. Then
\begin{equation*}
    \left(\mathbb{P}_n - \mathbb{P}\right) \left(\hat{f} - f\right) = O_{\mathbb{P}} \left(\frac{\|\hat{f} - f\|}{n^{1/2}}\right).
\end{equation*}
\end{lemma}

\begin{proof}
Denote the one-step estimator by
\begin{equation*}
    \hat{\Psi} = \Psi\left(\hat{\mathbb{P}}\right) + \mathbb{P}_n \left\{\phi^\ast_a\left(\hat{\mathbb{P}}\right)\right\},
\end{equation*}
which has the following decomposition by definition
\begin{align*}
\hat{\Psi} - \Psi &= \Psi\left(\hat{\mathbb{P}}\right) + \mathbb{P}_n \left\{\phi^\ast_a\left(\hat{\mathbb{P}}\right)\right\} - \Psi\left(\mathbb{P}\right) \\
&= \left(\mathbb{P}_n - \mathbb{P}\right) \left\{\phi^\ast_a\left(\mathbb{P}\right)\right\} + \left(\mathbb{P}_n - \mathbb{P}\right) \left\{\phi^\ast_a\left(\hat{\mathbb{P}}\right) - \phi^\ast_a\left(\mathbb{P}\right)\right\} + Z\left(\hat{\mathbb{P}}, \mathbb{P}\right),
\end{align*}
where $Z\left(\hat{\mathbb{P}}, \mathbb{P}\right) = \Psi\left(\hat{\mathbb{P}}\right) - \Psi\left(\mathbb{P}\right) + \mathbb{E}_{\mathbb{P}}\left[\phi^\ast_a\left(\hat{\mathbb{P}}\right)\right]$.

We analyze the three terms separately:
\begin{itemize}
    \item The first term
    \begin{equation*}
        \left(\mathbb{P}_n - \mathbb{P}\right) \left\{\phi^\ast_a\left(\mathbb{P}\right)\right\}
    \end{equation*}
    is a simple sample average, so by the central limit theorem, it converges to a normal random variable.

    \item The second term
    \begin{equation*}
        \left(\mathbb{P}_n - \mathbb{P}\right) \left\{\phi^\ast_a\left(\hat{\mathbb{P}}\right) - \phi^\ast_a\left(\mathbb{P}\right)\right\}
    \end{equation*}
    is a standard empirical process term. To show it is of the order $o_{\mathbb{P}}\left(n^{-1/2}\right)$, the classical approach is to assume $\phi^\ast_a$ belongs to a Donsker class. When machine learning algorithms are used, it is usually preferred to consider cross-fitting. Irrespective of the approach adopted, it is typically required that
    \begin{equation*}
        \left\|\phi^\ast_a\left(\hat{\mathbb{P}}\right) - \phi^\ast_a\left(\mathbb{P}\right)\right\|^2 = o_{\mathbb{P}}(1).
    \end{equation*}
    More specifically, in our analysis, a simple set of sufficient conditions includes the boundedness of the (estimated) nuisance functions and the following consistency:
    \begin{align*}
        \left\|\hat{f}(\Gamma = 1 \mid \boldsymbol{X}) - f(\Gamma = 1 \mid \boldsymbol{X})\right\| &= o_{\mathbb{P}}(1), \\
        \left\|\hat{f}(A = 1, \Gamma = 0 \mid \boldsymbol{X}) - f(A = 1, \Gamma = 0 \mid \boldsymbol{X})\right\| &= o_{\mathbb{P}}(1), \\
        \left\|\hat{f}(S \mid \boldsymbol{X}, A = a, \Gamma = 0) - f(S \mid \boldsymbol{X}, A = a, \Gamma = 0)\right\| &= o_{\mathbb{P}}(1), \\
        \left\|\hat{f}(A = a \mid \boldsymbol{X}, \Gamma = 1) - f(A = a \mid \boldsymbol{X}, \Gamma = 1)\right\| &= o_{\mathbb{P}}(1), \\
        \left\|\hat{\mu}(\boldsymbol{X}, a, S) - \mu(\boldsymbol{X}, a, S)\right\| &= o_{\mathbb{P}}(1),
    \end{align*}
    which together with Lemma~\ref{lemma:crossfitting} imply that
    \begin{equation*}
        \left(\mathbb{P}_n - \mathbb{P}\right) \left\{\phi^\ast_a\left(\hat{\mathbb{P}}\right) - \phi^\ast_a\left(\mathbb{P}\right)\right\} = o_{\mathbb{P}}\left(n^{-1/2}\right).
    \end{equation*}

    \item The third term
    \begin{equation*}
        Z\left(\hat{\mathbb{P}}, \mathbb{P}\right) = \Psi\left(\hat{\mathbb{P}}\right) - \Psi\left(\mathbb{P}\right) + \mathbb{E}_{\mathbb{P}}\left[\phi^\ast_a\left(\hat{\mathbb{P}}\right)\right]
    \end{equation*}
    is the key in our analysis. This remainder term is given by
    {\scriptsize
    \begin{align*}
        & Z\left(\hat{\mathbb{P}}, \mathbb{P}\right) \\
        & = \mathbb{E}\left[\left\{\frac{f(A = 1, \Gamma = 0 \mid \boldsymbol{X})}{\kappa \hat{f}(A = 1, \Gamma = 0 \mid \boldsymbol{X})} - \frac{1}{\kappa}\right\} \left\{\frac{\hat{f}(\Gamma = 1 \mid \boldsymbol{X}) \hat{f}(S \mid \boldsymbol{X}, A = a, \Gamma = 1)}{\hat{f}(S \mid \boldsymbol{X}, A = 1, \Gamma = 0)}- \frac{f(\Gamma = 1 \mid \boldsymbol{X}) f(S \mid \boldsymbol{X}, A = a, \Gamma = 1)}{f(S \mid \boldsymbol{X}, A = 1, \Gamma = 0)}\right\} \right. \\
        & \left. \qquad\quad \times \left\{\hat{\mu}(\boldsymbol{X}, a, S) - \mu(\boldsymbol{X}, a, S)\right\}\right. \\
        & \left. \quad + \left\{\frac{f(A = 1, \Gamma = 0 \mid \boldsymbol{X})}{\kappa \hat{f}(A = 1, \Gamma = 0 \mid \boldsymbol{X})} - \frac{1}{\kappa}\right\} \left\{\hat{\mu}(\boldsymbol{X}, a, S) - \mu(\boldsymbol{X}, a, S)\right\} \right. \\
        & \left. \quad + \resizebox{0.85\textwidth}{!}{$\left\{\frac{\hat{f}(\Gamma = 1 \mid \boldsymbol{X}) \hat{f}(S \mid \boldsymbol{X}, A = a, \Gamma = 1)}{\hat{f}(S \mid \boldsymbol{X}, A = 1, \Gamma = 0)}- \frac{f(\Gamma = 1 \mid \boldsymbol{X}) f(S \mid \boldsymbol{X}, A = a, \Gamma = 1)}{f(S \mid \boldsymbol{X}, A = 1, \Gamma = 0)}\right\} \left\{\hat{\mu}(\boldsymbol{X}, a, S) - \mu(\boldsymbol{X}, a, S)\right\}$} \right. \\
        & \left. \quad + \left\{\frac{f(A = a \mid \boldsymbol{X}, \Gamma = 1)}{\kappa \hat{f}(A = a \mid \boldsymbol{X}, \Gamma = 1)} - \frac{1}{\kappa}\right\} \left\{\hat{\mu}(\boldsymbol{X}, a, S) - \mu(\boldsymbol{X}, a, S)\right\} \right. \\
        & \left. \quad - \resizebox{0.85\textwidth}{!}{$\int \left\{\frac{f(A = a \mid \boldsymbol{X}, \Gamma = 1)}{\kappa \hat{f}(A = a \mid \boldsymbol{X}, \Gamma = 1)} - \frac{1}{\kappa}\right\} \left\{\hat{\mu}(\boldsymbol{X}, a, s) - \mu(\boldsymbol{X}, a, s)\right\} \left\{\hat{f}(s \mid  \boldsymbol{X}, A = a, \Gamma = 1) - f(s \mid  \boldsymbol{X}, A = a, \Gamma = 1)\right\} ds$} \right. \\
        & \left. \quad - \int \left\{\frac{f(A = a \mid \boldsymbol{X}, \Gamma = 1)}{\kappa \hat{f}(A = a \mid \boldsymbol{X}, \Gamma = 1)} - \frac{1}{\kappa}\right\} \mu(\boldsymbol{X}, a, s) \left\{\hat{f}(s \mid  \boldsymbol{X}, A = a, \Gamma = 1) - f(s \mid  \boldsymbol{X}, A = a, \Gamma = 1)\right\} ds \right. \\
        & \left. \quad - \int \left\{\frac{f(A = a \mid \boldsymbol{X}, \Gamma = 1)}{\kappa \hat{f}(A = a \mid \boldsymbol{X}, \Gamma = 1)} - \frac{1}{\kappa}\right\} \left\{\hat{\mu}(\boldsymbol{X}, a, s) - \mu(\boldsymbol{X}, a, s)\right\} f(s \mid  \boldsymbol{X}, A = a, \Gamma = 1) ds \right].
    \end{align*}}
    We assume the following rate of convergence conditions:
    {\scriptsize
    \begin{align*}
        \left\|\hat{\mu}(\boldsymbol{X}, a, S) - \mu(\boldsymbol{X}, a, S)\right\| \cdot \left\|\frac{f(A = 1, \Gamma = 0 \mid \boldsymbol{X})}{\hat{f}(A = 1, \Gamma = 0 \mid \boldsymbol{X})} - 1\right\| &= o_{\mathbb{P}}\left(n^{-1/2}\right), \\
        \resizebox{0.7\textwidth}{!}{$\left\|\hat{\mu}(\boldsymbol{X}, a, S) - \mu(\boldsymbol{X}, a, S)\right\| \cdot \left\|\frac{\hat{f}(\Gamma = 1 \mid \boldsymbol{X}) \hat{f}(S \mid \boldsymbol{X}, A = a, \Gamma = 1)}{\hat{f}(S \mid \boldsymbol{X}, A = 1, \Gamma = 0)}- \frac{f(\Gamma = 1 \mid \boldsymbol{X}) f(S \mid \boldsymbol{X}, A = a, \Gamma = 1)}{f(S \mid \boldsymbol{X}, A = 1, \Gamma = 0)}\right\|$} &= o_\mathbb{P}\left(n^{-1/2}\right), \\
        \left\|\hat{\mu}(\boldsymbol{X}, a, S) - \mu(\boldsymbol{X}, a, S)\right\| \cdot \left\|\frac{f(A = a \mid \boldsymbol{X}, \Gamma = 1)}{\hat{f}(A = a \mid \boldsymbol{X}, \Gamma = 1)} - 1\right\| &= o_\mathbb{P}\left(n^{-1/2}\right), \\
        \mathbb{E}\left|\int \left\{\frac{f(A = a \mid \boldsymbol{X}, \Gamma = 1)}{\kappa \hat{f}(A = a \mid \boldsymbol{X}, \Gamma = 1)} - \frac{1}{\kappa}\right\} \mu(\boldsymbol{X}, a, s) \left\{\hat{f}(s \mid  \boldsymbol{X}, A = a, \Gamma = 1) - f(s \mid  \boldsymbol{X}, A = a, \Gamma = 1)\right\} ds\right| &= o_\mathbb{P}\left(n^{-1/2}\right), \\
        \mathbb{E}\left|\int \left\{\frac{f(A = a \mid \boldsymbol{X}, \Gamma = 1)}{\kappa \hat{f}(A = a \mid \boldsymbol{X}, \Gamma = 1)} - \frac{1}{\kappa}\right\} \left\{\hat{\mu}(\boldsymbol{X}, a, s) - \mu(\boldsymbol{X}, a, s)\right\} f(s \mid  \boldsymbol{X}, A = a, \Gamma = 1) ds\right| &= o_\mathbb{P}\left(n^{-1/2}\right),
    \end{align*}}
    which imply 
    \begin{equation*}
        Z\left(\hat{\mathbb{P}}, \mathbb{P}\right) = o_{\mathbb{P}}\left(n^{-1/2}\right),
    \end{equation*}
    by Cauchy-Schwarz.

    Summarizing the above results, we have
    \begin{equation*}
        \hat{\Psi} - \Psi = \left(\mathbb{P}_n - \mathbb{P}\right) \left\{\phi^\ast_a\left(\mathbb{P}\right)\right\} + o_{\mathbb{P}}\left(n^{-1/2}\right),
    \end{equation*}
    which completes the proof.
    
\end{itemize}

\end{proof}

\subsection{Proof of Theorem~\ref{thm:EIF-censored-data} and \ref{thm:EIF-censored-data competing risks}}

\begin{proof}
When the outcome is subject to ignorable right censoring, the efficient influence function $\phi^{C\ast}_{a,t}$ can be derived from the complete data EIF $\phi^\ast_a$, where we set $y(T) = I\{T \leq t\}$ in Theorem \ref{thm:identification-censored-data}, via results given in \citet[Section~10.4]{tsiatis2006semiparametric}:
\begin{equation*}
\phi^{C\ast}_{a,t} = \frac{\Delta \phi^\ast_a}{G^{C}(\Tilde{T} \mid \boldsymbol{X}, A, S)} + \int_{0}^\infty \frac{L(u, \boldsymbol{X}, A, S)}{G^{C}(\Tilde{T} \mid \boldsymbol{X}, A, S)} dM^{C}(u, \boldsymbol{X}, A, S),
\end{equation*}
where $L(u, \boldsymbol{X}, A, S) = \mathbb{E}\left[\phi^\ast_a \mid T \geq u, \boldsymbol{X}, A, S\right]$, and we thus obtain $\phi^{C\ast}_{a,t}$ given by
{\small
\begin{align*}
& 1 - \frac{(1 - \Gamma) I\{A = 1\}}{\kappa} \frac{f(\Gamma = 1 \mid \boldsymbol{X})}{f(A = 1, \Gamma = 0 \mid \boldsymbol{X})} \frac{f(S \mid \boldsymbol{X}, A = a, \Gamma = 1)}{f(S \mid \boldsymbol{X}, A = 1, \Gamma = 0)} \left\{\frac{-I\{\Tilde{T} \leq t, \Delta = 1\} G^{T}(t \mid \boldsymbol{X}, A, S)}{G^{T}(\Tilde{T} \mid \boldsymbol{X}, A, S) G^{C}(\Tilde{T} \mid \boldsymbol{X}, A, S)} \right. \\
& \left.\quad + \int_{0}^{t \wedge \Tilde{T}} \frac{G^{T}(t \mid \boldsymbol{X}, A, S) \Lambda(du \mid \boldsymbol{X}, A, S)}{G^{T}(u \mid \boldsymbol{X}, A, S) G^{C}(u \mid \boldsymbol{X}, A, S)}\right\} \\
& - \frac{\Gamma I\{A = a\}}{\kappa f(A = a \mid \boldsymbol{X}, \Gamma = 1)} \left\{G^{T}(t \mid \boldsymbol{X}, a, S) - \int_{s\in \mathcal{S}} G^{T}(t \mid \boldsymbol{X}, a, s) f(s \mid \boldsymbol{X}, A = a, \Gamma = 1) ds\right\} \\
& - \frac{\Gamma}{\kappa} \int_{s\in \mathcal{S}} G^{T}(t \mid \boldsymbol{X}, a, s) f(s \mid  \boldsymbol{X}, A = a, \Gamma = 1) ds - R(a, t; \Gamma = 1).
\end{align*}}
The extension to censored competing risks data follows straightforwardly when we set $y(T) = I\{T \leq t, \Delta = j\}, j = 1, \ldots, J$ \citep{rytgaard2024targeted}.
\end{proof}

\subsection{Proof of Theorem~\ref{thm:MR-censored-data}}

\begin{proof}
Denote by $f^\ast$, $G^{T \ast}$ and $G^{C \ast}$ the probability limits of the nuisance function estimators $\hat{f}$, $\hat{G}^{T}$ and $\hat{G}^{C}$. Essentially we need to show
{\footnotesize
\begin{equation}\label{eq:MR-censored-data-proof}
\begin{split}
\mathbb{E} & \left[1 - \frac{(1 - \Gamma) I\{A = 1\}}{\kappa} \frac{f^{\ast}(\Gamma = 1 \mid \boldsymbol{X})}{f^{\ast}(A = 1, \Gamma = 0 \mid \boldsymbol{X})} \frac{f^{\ast}(S \mid \boldsymbol{X}, A = a, \Gamma = 1)}{f^{\ast}(S \mid \boldsymbol{X}, A = 1, \Gamma = 0)} \left\{\frac{-I\{\Tilde{T} \leq t, \Delta = 1\} G^{T \ast}(t \mid \boldsymbol{X}, A, S)}{G^{T \ast}(\Tilde{T} \mid \boldsymbol{X}, A, S) G^{C \ast}(\Tilde{T} \mid \boldsymbol{X}, A, S)} \right.\right. \\
& \left.\left.\quad + \int_{0}^{t \wedge \Tilde{T}} \frac{G^{T \ast}(t \mid \boldsymbol{X}, A, S) \Lambda^{\ast}(du \mid \boldsymbol{X}, A, S)}{G^{T \ast}(u \mid \boldsymbol{X}, A, S) G^{C \ast}(u \mid \boldsymbol{X}, A, S)}\right\} \right. \\
& \quad \left. - \frac{\Gamma I\{A = a\}}{\kappa f^{\ast}(A = a \mid \boldsymbol{X}, \Gamma = 1)} \left\{G^{T \ast}(t \mid \boldsymbol{X}, a, S) - \int_{s\in \mathcal{S}} G^{T \ast}(t \mid \boldsymbol{X}, a, s) f^{\ast}(s \mid \boldsymbol{X}, A = a, \Gamma = 1) ds\right\}\right. \\
& \quad \left. - \frac{\Gamma}{\kappa} \int_{s\in \mathcal{S}} G^{T \ast}(t \mid \boldsymbol{X}, a, s) f^{\ast}(s \mid  \boldsymbol{X}, A = a, \Gamma = 1) ds - R(a, t; \Gamma = 1)\right] = 0
\end{split}
\end{equation}}
under the union model $\mathcal{M}'_a \cup \mathcal{M}'_b \cup \mathcal{M}'_c$.

First note that $G^{T \ast}(t \mid \boldsymbol{x}, a, s) = G^{T}(t \mid \boldsymbol{x}, a, s)$ and $f^\ast(s \mid x, A = a, \Gamma = 1) = f(s \mid x, A = a, \Gamma = 1)$ under model $\mathcal{M}'_a$. Equation~\eqref{eq:MR-censored-data-proof} follows because
{\small
\begin{align*}
\mathbb{E} & \left[\frac{(1 - \Gamma) I\{A = 1\}}{\kappa} \frac{f^{\ast}(\Gamma = 1 \mid \boldsymbol{X})}{f^{\ast}(A = 1, \Gamma = 0 \mid \boldsymbol{X})} \frac{f(S \mid \boldsymbol{X}, A = a, \Gamma = 1)}{f(S \mid \boldsymbol{X}, A = 1, \Gamma = 0)} \left\{\frac{-I\{\Tilde{T} \leq t, \Delta = 1\} G^{T}(t \mid \boldsymbol{X}, A, S)}{G^{T}(\Tilde{T} \mid \boldsymbol{X}, A, S) G^{C \ast}(\Tilde{T} \mid \boldsymbol{X}, A, S)} \right.\right. \\
& \left.\left.\quad + \int_{0}^{t \wedge \Tilde{T}} \frac{G^{T}(t \mid \boldsymbol{X}, A, S) \Lambda(du \mid \boldsymbol{X}, A, S)}{G^{T}(u \mid \boldsymbol{X}, A, S) G^{C \ast}(u \mid \boldsymbol{X}, A, S)}\right\} \right] = 0,
\end{align*}}
\begin{equation*}
\mathbb{E}\left[\frac{\Gamma I\{A = a\}}{\kappa f^{\ast}(A = a \mid \boldsymbol{X}, \Gamma = 1)} \left\{G^{T}(t \mid \boldsymbol{X}, a, S) - \int_{s\in \mathcal{S}} G^{T}(t \mid \boldsymbol{X}, a, s) f(s \mid \boldsymbol{X}, A = a, \Gamma = 1) ds\right\}\right] = 0,
\end{equation*}
and
\begin{equation*}
\mathbb{E}\left[1 - \frac{\Gamma}{\kappa} \int_{s\in \mathcal{S}} G^{T}(t \mid \boldsymbol{X}, a, s) f(s \mid  \boldsymbol{X}, A = a, \Gamma = 1) ds\right] = R(a, t; \Gamma = 1).
\end{equation*}

Second, we have $G^{T \ast}(t \mid \boldsymbol{x}, a, s) = G^{T}(t \mid \boldsymbol{x}, a, s)$ and $f^\ast(A = a \mid x, \Gamma = 1) = f(A = a \mid x, \Gamma = 1)$ under model $\mathcal{M}_b$. Equation~\eqref{eq:MR-censored-data-proof} now follows because
{\small
\begin{align*}
\mathbb{E} & \left[\frac{(1 - \Gamma) I\{A = 1\}}{\kappa} \frac{f^{\ast}(\Gamma = 1 \mid \boldsymbol{X})}{f^{\ast}(A = 1, \Gamma = 0 \mid \boldsymbol{X})} \frac{f^{\ast}(S \mid \boldsymbol{X}, A = a, \Gamma = 1)}{f^{\ast}(S \mid \boldsymbol{X}, A = 1, \Gamma = 0)} \left\{\frac{-I\{\Tilde{T} \leq t, \Delta = 1\} G^{T}(t \mid \boldsymbol{X}, A, S)}{G^{T}(\Tilde{T} \mid \boldsymbol{X}, A, S) G^{C \ast}(\Tilde{T} \mid \boldsymbol{X}, A, S)} \right.\right. \\
& \left.\left.\quad + \int_{0}^{t \wedge \Tilde{T}} \frac{G^{T}(t \mid \boldsymbol{X}, A, S) \Lambda(du \mid \boldsymbol{X}, A, S)}{G^{T}(u \mid \boldsymbol{X}, A, S) G^{C \ast}(u \mid \boldsymbol{X}, A, S)}\right\} \right] = 0,
\end{align*}}
\begin{equation*}
\mathbb{E}\left[\left\{\frac{\Gamma}{\kappa} - \frac{\Gamma I\{A = a\}}{\kappa f(A = a \mid X, \Gamma = 1)}\right\} \int_{s\in \mathcal{S}} G^{T}(t \mid \boldsymbol{X}, a, s) f^{\ast}(s \mid  \boldsymbol{X}, A = a, \Gamma = 1) ds\right] = 0,
\end{equation*}
and
\begin{equation*}
\mathbb{E}\left[1 - \frac{\Gamma I\{A = a\}}{\kappa f(A = a \mid X, \Gamma = 1)} G^{T}(t \mid \boldsymbol{X}, a, s)\right] = R(a, t; \Gamma = 1).
\end{equation*}

Third, as $f^\ast(A = 1' \mid x, \Gamma = 1) = f(A = 1' \mid x, \Gamma = 1)$, $f^\ast(\Gamma = 1 \mid x) = f(\Gamma = 1 \mid x)$, $f^\ast(A = 1, \Gamma = 0 \mid x) = f(A = 1, \Gamma = 0 \mid x)$, $G^{C \ast}(u \mid \boldsymbol{x}, a, s) = G^{C}(u \mid \boldsymbol{x}, a, s)$, $f^\ast(s \mid x, A = 1', \Gamma = 1) = f(s \mid x, A = 1', \Gamma = 1)$ and $f^\ast(s \mid x, A = 1, \Gamma = 0) = f(s \mid x, A = 1, \Gamma = 0)$ under model $\mathcal{M}'_c$, we conclude that Equation~\eqref{eq:MR-censored-data-proof} holds because
{\small
\begin{align*}
& \mathbb{E}\left[1 - \frac{(1 - \Gamma) I\{A = 1\}}{\kappa} \frac{f(\Gamma = 1 \mid \boldsymbol{X})}{f(A = 1, \Gamma = 0 \mid \boldsymbol{X})} \frac{f(S \mid \boldsymbol{X}, A = a, \Gamma = 1)}{f(S \mid \boldsymbol{X}, A = 1, \Gamma = 0)} \left\{\frac{-I\{\Tilde{T} \leq t, \Delta = 1\} G^{T \ast}(t \mid \boldsymbol{X}, A, S)}{G^{T \ast}(\Tilde{T} \mid \boldsymbol{X}, A, S) G^{C}(\Tilde{T} \mid \boldsymbol{X}, A, S)} \right.\right. \\
& \left.\left.\quad\quad + \int_{0}^{t \wedge \Tilde{T}} \frac{G^{T \ast}(t \mid \boldsymbol{X}, A, S) \Lambda^{\ast}(du \mid \boldsymbol{X}, A, S)}{G^{T \ast}(u \mid \boldsymbol{X}, A, S) G^{C \ast}(u \mid \boldsymbol{X}, A, S)} - \frac{I\{\Tilde{T} \leq t, \Delta = 1\}}{G^{C}(\Tilde{T} \mid \mathbf{X}, A, S)}\right\}\right.\\
& \quad\quad \left. - \frac{\Gamma I\{A = a\}}{\kappa f(A = a \mid X, \Gamma = 1)} G^{T \ast}(t \mid \boldsymbol{X}, A, S)\right] \\
& = \mathbb{E}\left[\frac{(1 - \Gamma) I\{A = 1\}}{\kappa} \frac{f(\Gamma = 1 \mid \boldsymbol{X})}{f(A = 1, \Gamma = 0 \mid \boldsymbol{X})} \frac{f(S \mid \boldsymbol{X}, A = a, \Gamma = 1)}{f(S \mid \boldsymbol{X}, A = 1, \Gamma = 0)} G^{T \ast}(t \mid \boldsymbol{X}, A, S) \right.\\
& \left.\quad\quad\quad - \frac{\Gamma I\{A = a\}}{\kappa f(A = a \mid X, \Gamma = 1)} G^{T \ast}(t \mid \boldsymbol{X}, A, S)\right] = 0,
\end{align*}}
\begin{equation*}
\mathbb{E}\left[\left\{\frac{\Gamma}{\kappa} - \frac{\Gamma I\{A = a\}}{\kappa f(A = a \mid X, \Gamma = 1)}\right\} \int_{s\in \mathcal{S}} G^{T \ast}(t \mid \boldsymbol{X}, a, s) f(s \mid  \boldsymbol{X}, A = a, \Gamma = 1) ds \right] = 0,
\end{equation*}
and
\begin{equation*}
\mathbb{E}\left[\frac{(1 - \Gamma) I\{A = 1\}}{\kappa} \frac{f(\Gamma = 1 \mid X)}{f(A = 1, \Gamma = 0 \mid X)} \frac{f(S \mid X, A = a, \Gamma = 1)}{f(S \mid X, A = 1, \Gamma = 0)} \frac{\Delta I\{\Tilde{T} \leq t\}}{G^{C}(\Tilde{T} \mid \boldsymbol{X}, A, S)}\right] = R(a, t; \Gamma = 1),
\end{equation*}
which completes the proof of multiple robustness under right-censoring.
\end{proof}

\subsection{Proof of Proposition~\ref{prop: asymptotic linearity censored data}}\label{sec:SM proof asymptotic linearity censored data}

\begin{proof}
For ease of notation, denote the one-step estimator by
\begin{equation*}
    \hat{\Psi}_{t} = \Psi_{t}\left(\hat{\mathbb{P}}\right) + \mathbb{P}_n \left\{\phi^{C\ast}_{a,t}\left(\hat{\mathbb{P}}\right)\right\},
\end{equation*}
which has the following decomposition by definition
\begin{align*}
\hat{\Psi}_{t} - \Psi_{t} &= \Psi_{t}\left(\hat{\mathbb{P}}\right) + \mathbb{P}_n \left\{\phi^{C\ast}_{a,t}\left(\hat{\mathbb{P}}\right)\right\} - \Psi_{t}\left(\mathbb{P}\right) \\
&= \left(\mathbb{P}_n - \mathbb{P}\right) \left\{\phi^{C\ast}_{a,t}\left(\mathbb{P}\right)\right\} + \left(\mathbb{P}_n - \mathbb{P}\right) \left\{\phi^{C\ast}_{a,t}\left(\hat{\mathbb{P}}\right) - \phi^{C\ast}_{a,t}\left(\mathbb{P}\right)\right\} + Z\left(\hat{\mathbb{P}}, \mathbb{P}\right),
\end{align*}
where $Z\left(\hat{\mathbb{P}}, \mathbb{P}\right) = \Psi_{t}\left(\hat{\mathbb{P}}\right) - \Psi_{t}\left(\mathbb{P}\right) + \mathbb{E}_{\mathbb{P}}\left[\phi^{C\ast}_{a,t}\left(\hat{\mathbb{P}}\right)\right]$.

We analyze the three terms separately:
\begin{itemize}
    \item The first term
    \begin{equation*}
        \left(\mathbb{P}_n - \mathbb{P}\right) \left\{\phi^{C\ast}_{a,t}\left(\mathbb{P}\right)\right\}
    \end{equation*}
    is a simple sample average, so by the central limit theorem, it converges to a normal random variable.

    \item The second term
    \begin{equation*}
        \left(\mathbb{P}_n - \mathbb{P}\right) \left\{\phi^{C\ast}_{a,t}\left(\hat{\mathbb{P}}\right) - \phi^{C\ast}_{a,t}\left(\mathbb{P}\right)\right\}
    \end{equation*}
    is a standard empirical process term. It is typically required that
    \begin{equation*}
        \left\|\phi^{C\ast}_{a,t}\left(\hat{\mathbb{P}}\right) - \phi^{C\ast}_{a,t}\left(\mathbb{P}\right)\right\|^2 = o_{\mathbb{P}}(1).
    \end{equation*}
    More specifically, in our analysis, a simple set of sufficient conditions includes the boundedness of the (estimated) nuisance functions and the following consistency:
    \begin{align*}
        \left\|\hat{f}(\Gamma = 1 \mid \boldsymbol{X}) - f(\Gamma = 1 \mid \boldsymbol{X})\right\| &= o_{\mathbb{P}}(1), \\
        \left\|\hat{f}(A = 1, \Gamma = 0 \mid \boldsymbol{X}) - f(A = 1, \Gamma = 0 \mid \boldsymbol{X})\right\| &= o_{\mathbb{P}}(1), \\
        \left\|\hat{f}(S \mid \boldsymbol{X}, A = a, \Gamma = 0) - f(S \mid \boldsymbol{X}, A = a, \Gamma = 0)\right\| &= o_{\mathbb{P}}(1), \\
        \left\|\hat{f}(A = a \mid \boldsymbol{X}, \Gamma = 1) - f(A = a \mid \boldsymbol{X}, \Gamma = 1)\right\| &= o_{\mathbb{P}}(1), \\
        \left\|\sup_{u \in [0, t]} \left|\hat{G}^{C}(u \mid \boldsymbol{x}, a, s) - G^{C}(u \mid \boldsymbol{x}, a, s)\right|\right\| &= o_{\mathbb{P}}(1), \\
        \left\|\sup_{u \in [0, t]} \left|\frac{\hat{G}^{T}(t \mid \boldsymbol{x}, a, s)}{\hat{G}^{T}(u \mid \boldsymbol{x}, a, s)} - \frac{G^{T}(t \mid \boldsymbol{x}, a, s)}{G^{T}(u \mid \boldsymbol{x}, a, s)}\right|\right\| &= o_{\mathbb{P}}(1),
    \end{align*}
    which together with Lemma~\ref{lemma:crossfitting} imply that
    \begin{equation*}
        \left(\mathbb{P}_n - \mathbb{P}\right) \left\{\phi^{C\ast}_{a,t}\left(\hat{\mathbb{P}}\right) - \phi^{C\ast}_{a,t}\left(\mathbb{P}\right)\right\} = o_{\mathbb{P}}\left(n^{-1/2}\right).
    \end{equation*}
    Note that for uniform inference, it is also required
    \begin{equation*}
        \left\|\sup_{u \in [0, t]}\sup_{v \in [0, u]} \left|\frac{\hat{G}^{T}(u \mid \boldsymbol{x}, a, s)}{\hat{G}^{T}(v \mid \boldsymbol{x}, a, s)} - \frac{G^{T}(u \mid \boldsymbol{x}, a, s)}{G^{T}(v \mid \boldsymbol{x}, a, s)}\right|\right\| = o_{\mathbb{P}}(1).
    \end{equation*}

    \item The third term
    \begin{equation*}
        Z\left(\hat{\mathbb{P}}, \mathbb{P}\right) = \Psi_{t}\left(\hat{\mathbb{P}}\right) - \Psi_{t}\left(\mathbb{P}\right) + \mathbb{E}_{\mathbb{P}}\left[\phi^{C\ast}_{a,t}\left(\hat{\mathbb{P}}\right)\right]
    \end{equation*}
    is the key in our analysis. This remainder term is given by
    {\scriptsize
    \begingroup
    \allowdisplaybreaks
    \begin{align*}
        & Z\left(\hat{\mathbb{P}}, \mathbb{P}\right) \\
        & = \mathbb{E}\left[\left\{\frac{f(A = 1, \Gamma = 0 \mid \boldsymbol{X})}{\kappa \hat{f}(A = 1, \Gamma = 0 \mid \boldsymbol{X})} - \frac{1}{\kappa}\right\} \left\{\hat{G}^{T}(t \mid \boldsymbol{X}, a, S) - G^{T}(t \mid \boldsymbol{X}, a, S)\right\}\right. \\
        & \left. \qquad\quad \times \left\{\frac{\hat{f}(\Gamma = 1 \mid \boldsymbol{X}) \hat{f}(S \mid \boldsymbol{X}, A = a, \Gamma = 1)}{\hat{f}(S \mid \boldsymbol{X}, A = 1, \Gamma = 0)}- \frac{f(\Gamma = 1 \mid \boldsymbol{X}) f(S \mid \boldsymbol{X}, A = a, \Gamma = 1)}{f(S \mid \boldsymbol{X}, A = 1, \Gamma = 0)}\right\} \right. \\
        & \left. \quad + \left\{\frac{f(A = 1, \Gamma = 0 \mid \boldsymbol{X})}{\kappa \hat{f}(A = 1, \Gamma = 0 \mid \boldsymbol{X})} - \frac{1}{\kappa}\right\} \left\{\hat{G}^{T}(t \mid \boldsymbol{X}, a, S) - G^{T}(t \mid \boldsymbol{X}, a, S)\right\} \right. \\
        & \left. \quad + \left\{\frac{\hat{f}(\Gamma = 1 \mid \boldsymbol{X}) \hat{f}(S \mid \boldsymbol{X}, A = a, \Gamma = 1)}{\hat{f}(S \mid \boldsymbol{X}, A = 1, \Gamma = 0)}- \frac{f(\Gamma = 1 \mid \boldsymbol{X}) f(S \mid \boldsymbol{X}, A = a, \Gamma = 1)}{f(S \mid \boldsymbol{X}, A = 1, \Gamma = 0)}\right\} \right. \\
        & \left. \qquad\quad \times \left\{\hat{G}^{T}(t \mid \boldsymbol{X}, a, S) - G^{T}(t \mid \boldsymbol{X}, a, S)\right\} \right. \\
        & \left. \quad + \left\{\frac{f(A = a \mid \boldsymbol{X}, \Gamma = 1)}{\kappa \hat{f}(A = a \mid \boldsymbol{X}, \Gamma = 1)} - \frac{1}{\kappa}\right\} \left\{\hat{G}^{T}(t \mid \boldsymbol{X}, a, S) - G^{T}(t \mid \boldsymbol{X}, a, S)\right\} \right. \\
        & \left. \quad + \frac{f(A = 1, \Gamma = 0 \mid \boldsymbol{X}) f(\Gamma = 1 \mid \boldsymbol{X}) f(S \mid \boldsymbol{X}, A = a, \Gamma = 1)}{\hat{f}(A = 1, \Gamma = 0 \mid \boldsymbol{X}) f(S \mid \boldsymbol{X}, A = 1, \Gamma = 0)} \right. \\
        & \left. \quad\quad\quad \times \hat{G}^{T}(t \mid \boldsymbol{X}, a, S) \int_{0}^{t} \left\{\frac{G^{C}(u \mid \boldsymbol{X}, a, S)}{\hat{G}^{C}(u \mid \boldsymbol{X}, a, S)} - 1\right\} \left(\frac{G^{T}}{\hat{G}^{T}} - 1\right) (du \mid \boldsymbol{X}, a, S) \right. \\
        & \left. \quad - \int \left\{\frac{f(A = a \mid \boldsymbol{X}, \Gamma = 1)}{\kappa \hat{f}(A = a \mid \boldsymbol{X}, \Gamma = 1)} - \frac{1}{\kappa}\right\} \left\{\hat{G}^{T}(t \mid \boldsymbol{X}, a, s) - G^{T}(t \mid \boldsymbol{X}, a, s)\right\} \right. \\
        & \left. \qquad\qquad \times \left\{\hat{f}(s \mid  \boldsymbol{X}, A = a, \Gamma = 1) - f(s \mid  \boldsymbol{X}, A = a, \Gamma = 1)\right\} ds \right. \\
        & \left. \quad - \int \left\{\frac{f(A = a \mid \boldsymbol{X}, \Gamma = 1)}{\kappa \hat{f}(A = a \mid \boldsymbol{X}, \Gamma = 1)} - \frac{1}{\kappa}\right\} G^{T}(t \mid \boldsymbol{X}, a, s) \left\{\hat{f}(s \mid  \boldsymbol{X}, A = a, \Gamma = 1) - f(s \mid  \boldsymbol{X}, A = a, \Gamma = 1)\right\} ds \right. \\
        & \left. \quad - \int \left\{\frac{f(A = a \mid \boldsymbol{X}, \Gamma = 1)}{\kappa \hat{f}(A = a \mid \boldsymbol{X}, \Gamma = 1)} - \frac{1}{\kappa}\right\} \left\{\hat{G}^{T}(t \mid \boldsymbol{X}, a, s) - G^{T}(t \mid \boldsymbol{X}, a, s)\right\} f(s \mid  \boldsymbol{X}, A = a, \Gamma = 1) ds \right].
    \end{align*}
    \endgroup}
    We assume the following rate of convergence conditions:
    {\scriptsize
    \begin{align*}
        \left\|\hat{G}^{T}(t \mid \boldsymbol{X}, a, S) - G^{T}(t \mid \boldsymbol{X}, a, S)\right\| \cdot \left\|\frac{f(A = 1, \Gamma = 0 \mid \boldsymbol{X})}{\hat{f}(A = 1, \Gamma = 0 \mid \boldsymbol{X})} - 1\right\| &= o_{\mathbb{P}}\left(n^{-1/2}\right), \\
        \resizebox{0.7\textwidth}{!}{$\left\|\hat{G}^{T}(t \mid \boldsymbol{X}, a, S) - G^{T}(t \mid \boldsymbol{X}, a, S)\right\| \cdot \left\|\frac{\hat{f}(\Gamma = 1 \mid \boldsymbol{X}) \hat{f}(S \mid \boldsymbol{X}, A = a, \Gamma = 1)}{\hat{f}(S \mid \boldsymbol{X}, A = 1, \Gamma = 0)}- \frac{f(\Gamma = 1 \mid \boldsymbol{X}) f(S \mid \boldsymbol{X}, A = a, \Gamma = 1)}{f(S \mid \boldsymbol{X}, A = 1, \Gamma = 0)}\right\|$} &= o_\mathbb{P}\left(n^{-1/2}\right), \\
        \mathbb{E}\left|\hat{G}^{T}(t \mid \boldsymbol{X}, a, S) \int_{0}^{t} \left\{\frac{G^{C}(u \mid \boldsymbol{X}, a, S)}{\hat{G}^{C}(u \mid \boldsymbol{X}, a, S)} - 1\right\} \left(\frac{G^{T}}{\hat{G}^{T}} - 1\right) (du \mid \boldsymbol{X}, a, S)\right| &= o_\mathbb{P}\left(n^{-1/2}\right), \\
        \left\|\hat{G}^{T}(t \mid \boldsymbol{X}, a, S) - G^{T}(t \mid \boldsymbol{X}, a, S)\right\| \cdot \left\|\frac{f(A = a \mid \boldsymbol{X}, \Gamma = 1)}{\hat{f}(A = a \mid \boldsymbol{X}, \Gamma = 1)} - 1\right\| &= o_\mathbb{P}\left(n^{-1/2}\right), \\
        \resizebox{0.7\textwidth}{!}{$\mathbb{E}\left|\int \left\{\frac{f(A = a \mid \boldsymbol{X}, \Gamma = 1)}{\kappa \hat{f}(A = a \mid \boldsymbol{X}, \Gamma = 1)} - \frac{1}{\kappa}\right\} \left\{\hat{G}^{T}(t \mid \boldsymbol{X}, a, s) - G^{T}(t \mid \boldsymbol{X}, a, s)\right\} \left\{\hat{f}(s \mid  \boldsymbol{X}, A = a, \Gamma = 1) - f(s \mid  \boldsymbol{X}, A = a, \Gamma = 1)\right\} ds\right|$} &= o_\mathbb{P}\left(n^{-1/2}\right), \\
        \resizebox{0.7\textwidth}{!}{$\mathbb{E}\left|\int \left\{\frac{f(A = a \mid \boldsymbol{X}, \Gamma = 1)}{\kappa \hat{f}(A = a \mid \boldsymbol{X}, \Gamma = 1)} - \frac{1}{\kappa}\right\} G^{T}(t \mid \boldsymbol{X}, a, s) \left\{\hat{f}(s \mid  \boldsymbol{X}, A = a, \Gamma = 1) - f(s \mid  \boldsymbol{X}, A = a, \Gamma = 1)\right\} ds\right|$} &= o_\mathbb{P}\left(n^{-1/2}\right), \\
        \resizebox{0.7\textwidth}{!}{$\mathbb{E}\left|\int \left\{\frac{f(A = a \mid \boldsymbol{X}, \Gamma = 1)}{\kappa \hat{f}(A = a \mid \boldsymbol{X}, \Gamma = 1)} - \frac{1}{\kappa}\right\} \left\{\hat{G}^{T}(t \mid \boldsymbol{X}, a, s) - G^{T}(t \mid \boldsymbol{X}, a, s)\right\} f(s \mid  \boldsymbol{X}, A = a, \Gamma = 1) ds\right|$} &= o_\mathbb{P}\left(n^{-1/2}\right),
    \end{align*}}
    which imply 
    \begin{equation*}
        Z\left(\hat{\mathbb{P}}, \mathbb{P}\right) = o_{\mathbb{P}}\left(n^{-1/2}\right),
    \end{equation*}
    by Cauchy-Schwarz. Note that for uniform inference, it is required that these conditions should hold uniformly over $t$.

    Summarizing the above results, we have
    \begin{equation*}
        \hat{\Psi}_{t} - \Psi_{t} = \left(\mathbb{P}_n - \mathbb{P}\right) \left\{\phi^{C\ast}_{a,t}\left(\mathbb{P}\right)\right\} + o_{\mathbb{P}}\left(n^{-1/2}\right),
    \end{equation*}
    which completes the proof.
    
\end{itemize}

\end{proof}

%% file: SM.Data.tex
\clearpage
\section{Additional details on the case study}
\renewcommand{\thefigure}{S\arabic{figure}}
\setcounter{figure}{0}

\begin{figure}[ht]
    \centering
\includegraphics[width=0.85\linewidth]{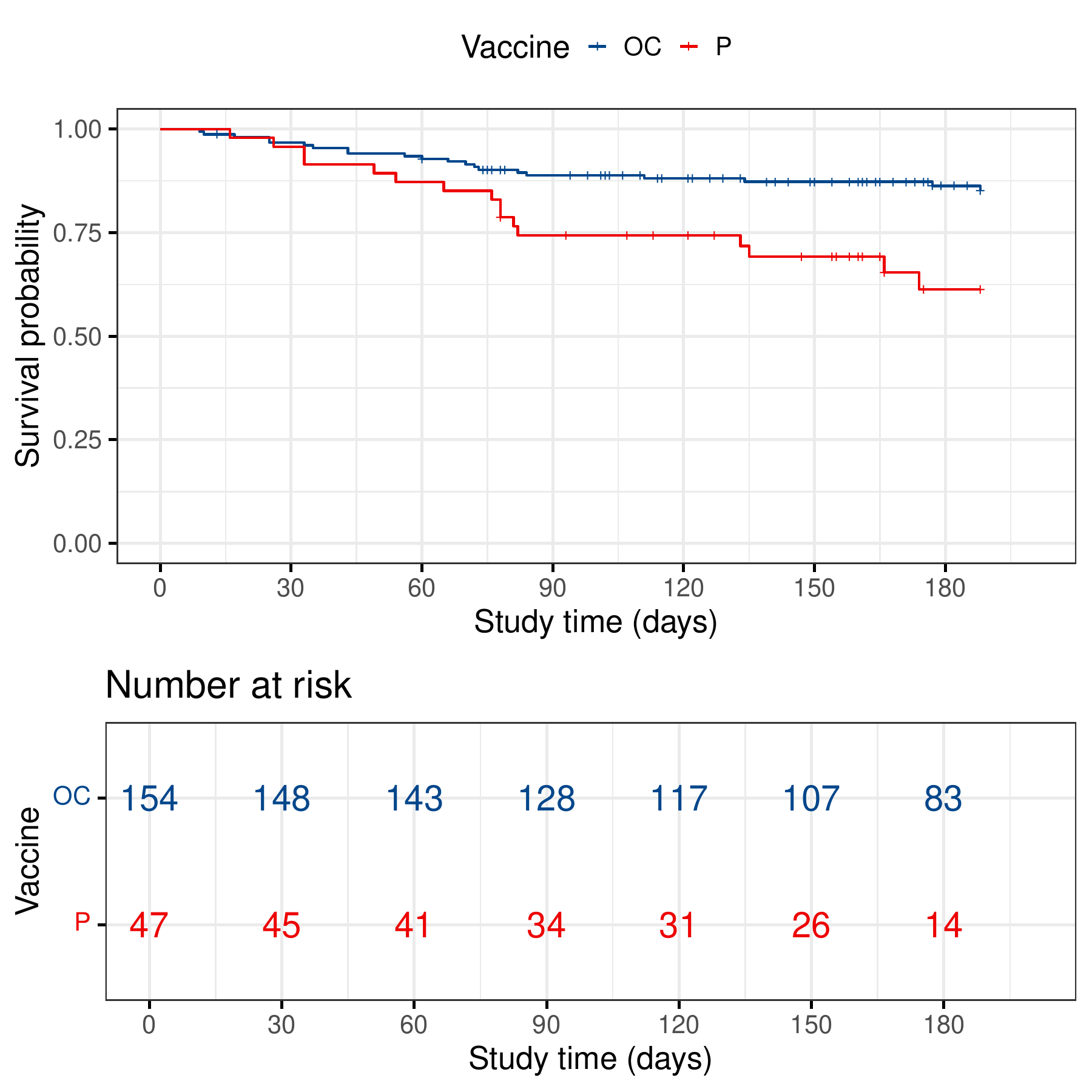}
    \caption{The Kaplan-Meier curves and associated risk tables for Stage-2 Omicron-containing (OC) or Prototype (P) vaccine recipients.}
    \label{fig: case study stage 2 ggsurvplot with risk table}
\end{figure}

\begin{figure}[ht]
    \centering
\includegraphics[width=0.85\linewidth]{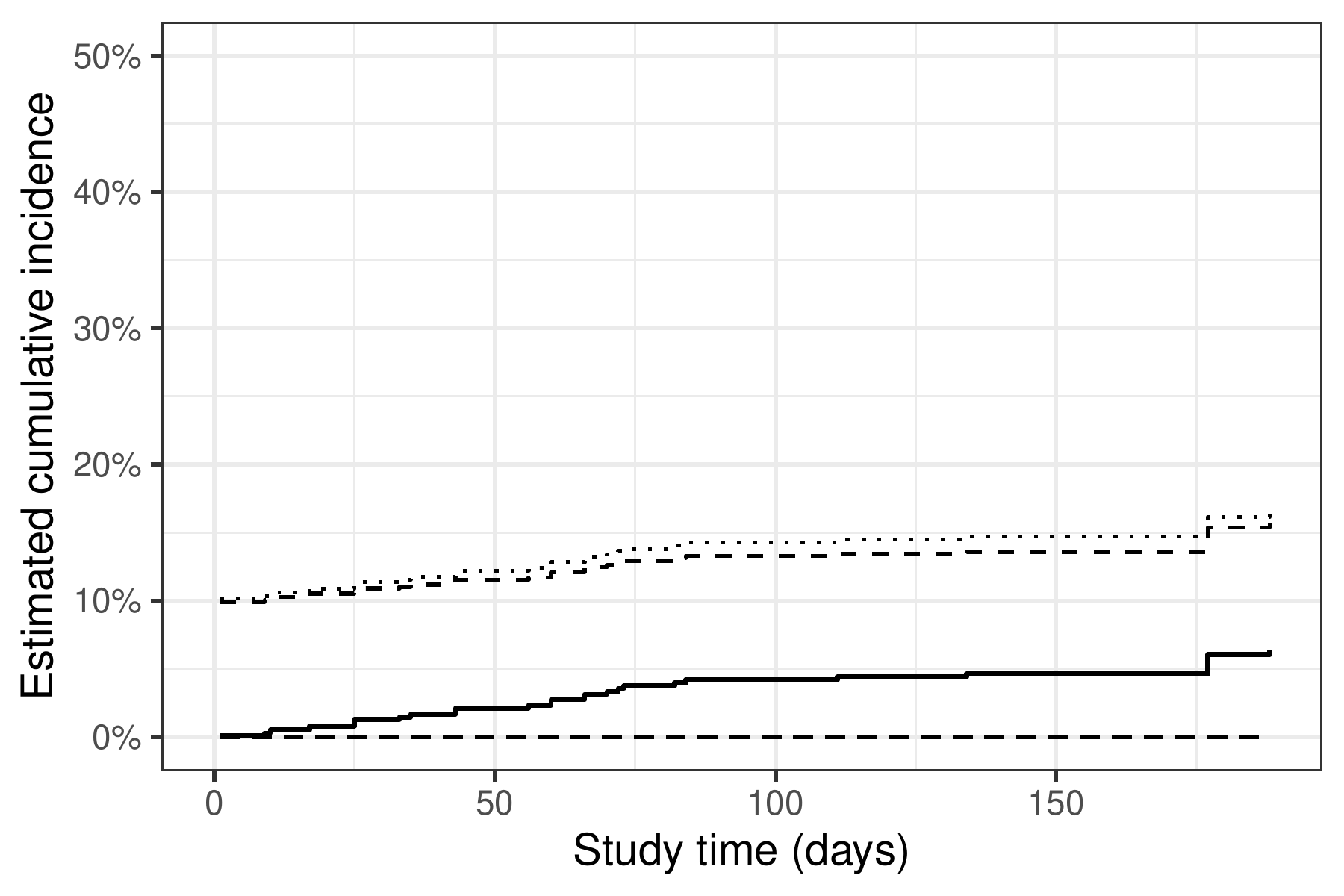}
    \caption{The cumulative incidence curve for Stage-4 BA.4/5 + Prototype Pfizer--BioNTech vaccine recipients estimated using the immunogenicity and clinical data of Stage-2 Omicron-containing Pfizer--BioNTech vaccine recipients ($\mathcal{D}_h$) as well as the immunogenicity data of Stage-4  BA.4/5 + Prototype Pfizer--BioNTech vaccine recipients ($\mathcal{D}_b$). Dashed lines indicate 95\% pointwise confidence intervals, and dotted lines indicate 95\% uniformly valid confidence bands.}
    \label{fig: case study stage 2 actual and projected incidence uniform}
\end{figure}